\documentclass[9pt,shortpaper,twoside,web]{ieeecolor}

\usepackage{times}
\usepackage{setspace}
\usepackage{url}
\spacing{1}
\usepackage[utf8]{inputenc}
\usepackage[T1]{fontenc}
\usepackage{graphicx}		
\usepackage{wrapfig}
\usepackage{sidecap} 
\usepackage[export]{adjustbox}
\usepackage{subcaption}
\usepackage[font=small]{caption}
\usepackage{float}

\usepackage{amsmath} 
\usepackage{amssymb}  
\usepackage{amsthm}
\usepackage{mathtools}
\usepackage[normalem]{ulem}
\usepackage{paralist}	
\usepackage[space]{grffile} 
\usepackage{color}
\let\labelindent\relax
\usepackage{enumitem}
\usepackage{bm}
\usepackage{cancel}
\usepackage{hhline}
\usepackage[c2 , nocomma]{optidef}
\usepackage{oubraces}
\usepackage{algorithmic}
\usepackage{graphicx}
\usepackage{textcomp}

\usepackage{amsfonts}
\usepackage{lipsum}
\usepackage{cite}

\newtheorem{theorem}{Theorem}

\newtheorem{lemma}{Lemma}
\theoremstyle{definition}
\newtheorem{definition}{Definition}
\theoremstyle{remark}
\newtheorem{remark}{Remark}
\theoremstyle{definition}
\newtheorem{assumption}{Assumption}
\theoremstyle{definition}

\newtheorem{problem}{Problem}

\newcommand{\C}{\mathcal{C}}

\definecolor{darkblue}{RGB}{0,0,102}
\definecolor{lightblue}{RGB}{77,77,148}

\definecolor{gold}{RGB}{234, 170, 0}
\definecolor{metallic_gold}{RGB}{139, 111, 78}

\newcommand{\norm}[1]{\left\Vert #1 \right\Vert}

\newcommand{\derp}[2]{\frac{\partial #1 }{\partial #2 }}

\newcommand{\grad}{\nabla}

\DeclareMathOperator{\diag}{diag}
\DeclareMathOperator{\rank}{rank}

\DeclareMathOperator{\Ima}{Im}

\DeclareMathOperator*{\argmin}{arg\,min}

\usepackage{xfrac}

\usepackage{dblfloatfix}

\usepackage{pdfpages}
\usepackage{graphicx}

\usepackage{enumitem}

\usepackage{comment}
\usepackage{generic}
\renewcommand{\baselinestretch}{0.93} 

\def\BibTeX{{\rm B\kern-.05em{\sc i\kern-.025em b}\kern-.08em
    T\kern-.1667em\lower.7ex\hbox{E}\kern-.125emX}}

\begin{document}

\title{ \bf
Robust Control Barrier Functions using Uncertainty Estimation \\ with Application to Mobile Robots
}

\author{Ersin Da\c{s}, \IEEEmembership{Member, IEEE}, Joel W. Burdick, \IEEEmembership{Member, IEEE}
\thanks{*This work was supported by DARPA under the LINC program.}
\thanks{The authors are with the Dept. of Mechanical and Civil Engineering, Caltech, Pasadena, CA 91125. ${\tt\small \{ersindas, jburdick \}@caltech.edu}$ } }

\maketitle
\thispagestyle{empty}     
\pagestyle{empty}

\begin{abstract}
This paper proposes a safety-critical control design approach for nonlinear control affine systems in the presence of matched and unmatched uncertainties. Our constructive framework couples control barrier function (CBF) theory with a new uncertainty estimator to ensure robust safety. We use the estimated uncertainty, along with a derived upper bound on the estimation error, for synthesizing CBFs and safety-critical controllers via a quadratic program-based feedback control law that rigorously ensures robust safety while improving disturbance rejection performance. We extend the method to higher-order CBFs (HOCBFs) to achieve safety under unmatched uncertainty, which may cause relative degree differences with respect to control input and disturbances. We assume the relative degree difference is at most one, resulting in a second-order cone constraint. We demonstrate the proposed robust HOCBF method through a simulation of an uncertain elastic actuator control problem and experimentally validate the efficacy of our robust CBF framework on a tracked robot with slope-induced matched and unmatched perturbations. 
\end{abstract}

\begin{IEEEkeywords}
Control barrier functions, constrained control, robotics, robust control, uncertainty estimation
\end{IEEEkeywords}

\section{Introduction} 
\label{sec:intro}
New applications of modern autonomous control systems in complex environments increase the importance of system safety. For example, mobile robots, autonomous vehicles, and robot manipulators operating in the presence of humans must be controlled safely. However, safety-critical controller design is a challenging problem in the deployment of real-world autonomous and cyber-physical systems \cite{guiochet2017}. These problems motivate the need for a control design that provides theoretical safety assurances.

Control system safety is often based on set invariance, which ensures that system states always remain within a safe set. Methods to synthesize controllers that guarantee safe set forward invariance include predictive control-based techniques \cite{hewing2020learning}, reachability-based methods \cite{bansal2017}, and control barrier functions (CBFs) \cite{ames2017control}. These frameworks also yield safety filters that minimally alter control inputs when safety risks arise. For nonlinear control affine systems, CBF-based schemes can enforce safety conditions via quadratic programs. CBF effectiveness has been shown in various applications \cite{wabersich2023data}.

Conventional CBF formulations may be sensitive to inevitable model uncertainties or external disturbances. This sensitivity triggers the need for assessing safety under unmodelled system dynamics \cite{xu2015robustness}. Several works have proposed robustifying terms in the CBF condition to address the robustness against disturbances and model uncertainties \cite{alan2023parameterized}. Specifically, \cite{jankovic2018robust} and \cite{kolathaya2018CLF} use a constant bound to represent the unmodeled dynamics in CBF constraints. This constant bound may be difficult to tune in practice, and usually results in undesired conservativeness and a reduction in closed-loop control performance. Tunable Input-to-State Safe CBF (TISSf-CBF) \cite{alan2021safe} has been proposed to reduce the conservatism of these methods.

Robust safety-critical control systems must balance safety and performance. Maintaining safety is crucial, but avoiding unnecessary conservativeness that negatively impacts performance is also essential \cite{bikas2021prescribed, bechlioulis2008robust}. Several safe, robust adaptive control design methods using CBFs \cite{lopez2023unmatched, black2021fixed} have been proposed to increase robust controller performance by using online estimation under the assumption of parametric system uncertainties. Recently, machine learning methods have been combined with adaptive CBFs \cite{cohen2023adaptive}.

Recent works have connected disturbance observers from robust control with CBFs to ensure robust safety with the active estimation of an external disturbance \cite{zhao2020adaptive, dacs2022robust, alan2022disturbance, wang2022disturbance, rise}. These prior studies have been limited to state-dependent external disturbances. However, uncertainties may depend on both system states and controls. Nor did these studies use active disturbance/uncertainty compensators to design observers that improve control performance while guaranteeing safety.  

It may be difficult to define a CBF when its time derivative does not depend on the control input. High-order CBFs (HOCBFs) \cite{xiao2019control, tan2021high} address this limitation. However, a CBF may have different relative degrees with respect to the input and unmatched uncertainties. If the input relative degree (IRD) is greater than the disturbance relative degree (DRD), the time derivatives of the uncertainties, which are naturally unknown, appear in the low-order CBF derivatives that are used in HOCBF methods with dynamical extensions \cite{takano2020robust}. The extension of HOCBFs to address robustness against unmatched uncertainties has not yet been considered.

Motivated by the aforementioned limitations, we propose more general uncertainty/disturbance estimation-based robust CBF and HOCBF frameworks for control-affine nonlinear systems subject to external disturbances and time-varying, state and input-dependent matched and unmatched uncertainties. Our main contributions are: 
\begin{enumerate}
  \item We connect existing CBF conditions with uniformly ultimately bounded uncertainty estimators to develop a robust safety-critical control scheme. We also compose the estimation and control law with a CBF-QP output to improve robust control performance, while guaranteeing robust safety, by compensating the matched uncertainty. We extend our method to HOCBFs, achieving safety under unmatched uncertainty, which may pose relative degree differences between the control inputs and disturbances. When ${\text{DRD} \!-\! \text{IRD} \!=\!\! 1}$, our method yields a second-order cone program.
  \item We introduce a new uncertainty/disturbance estimator that extends the observer studied in \cite{stotsky2002, dacs2022robust, alan2022disturbance} to observe unmodelled dynamics in nonlinear systems with multiple disturbance inputs. An upper bound for the estimation error is developed under boundedness assumptions on the uncertainty.
  \item We experimentally demonstrate our CBF method on a tracked mobile robot, where both unmatched and matched uncertainties occur. We showcase our robust HOCBF results via a simulation of an elastic actuator control problem in the presence of unmatched and matched uncertainties, which cause relative degree differences with respect to input and disturbance. 
\end{enumerate}

We organize this paper as follows. Section~\ref{sec:pre} provides preliminaries. Section~\ref{sec:main} proposes and analyzes an uncertainty estimator and introduces estimator-based robust CBF schemes. Section~\ref{ex:sim_hard} presents simulations and experiments, while Section~\ref{sec:conc} concludes the paper.

\section{Preliminaries}
\label{sec:pre}
\textbf{\textit{Notation:}} ${\mathbb{N}, \mathbb{R}, \mathbb{R}^+, \mathbb{R}^+_0}$ represent the set of natural, real, positive real, and non-negative real numbers, respectively. The Euclidean (the Frobenius norm) norm of a matrix is denoted by $\|\!\cdot\!\|$. A zero vector is denoted by ${\bf 0}$. ${\partial \C}$ denotes the boundary of a set ${\C}$. 

We consider nonlinear control affine systems of the form:
\begin{equation}
\label{system}
    \dot{x}  = f(x) + g(x) u,
\end{equation}
where ${x \!\in\! X \!\subset\! \mathbb{R}^n}$ is the state, ${u \!\in\! U \!\subset\! \mathbb{R}^m}$ is the control input, and ${f\!:\! X \!\to\! \mathbb{R}^n}$, ${g\!:\! X \!\to\! \mathbb{R}^{n \times m} }$ are locally Lipschitz continuous functions on $X$. We call \eqref{system} the \textit{actual (or uncertain) model}. A locally Lipschitz continuous controller ${u \!=\! \mathbf{k}(x)}$, with ${\mathbf{k}\!:\! X \!\to\! U}$, yields a locally Lipschitz continuous \textit{closed-loop} control system, ${f_{\rm cl}\!:\! X \!\to\! \mathbb{R}^n}$:
\begin{equation}
\label{eq:clsystem1}
    \dot{x} = {f}(x) + {g}(x) \mathbf{k}(x) \triangleq f_{\rm cl}(x).
\end{equation}
Hence, given any initial condition ${x_0 \!\triangleq\! x(t_0) \!\in\! X}$ there exists an interval ${\mathcal{I} ({x}_{0} ) \!\triangleq\! \left[t_0, t_{\max }\right)}$ such that the solution of \eqref{eq:clsystem1} produces a unique trajectory for ${t \!\in\! \mathcal{I} ({x}_{0} )}$. Throughout this study, we assume ${f_{\rm cl}}$ is forward complete, i.e., ${\mathcal{I} ({x}_{0} ) \!=\! [0, \infty)}$, and ${U}$ is a convex polytope. 

\subsection{Control Barrier Functions} 
\label{sec:cbfs}
We consider a set ${\C \!\subset\! X }$ defined as the 0-superlevel set of a continuously differentiable function ${h\!:\! X \!\to\! \mathbb{R}}$ as
\begin{equation}
\label{eq:CBF1}
    \C \!\triangleq\! \left\{ x \!\in\! X \!\subset\! \mathbb{R}^n \!:\! h(x) \!\geq\! 0 \right\}, \
    \partial \C \!\triangleq\! \left\{ {x \!\in\! X \!\subset\! \mathbb{R}^n} \!:\! h(x) \!=\! 0 \right\}. 
\end{equation}
This set is \textit{forward invariant} if, for every initial condition ${x(0) \!\in\! \C}$, the solution of \eqref{eq:clsystem1} satisfies ${x(t) \!\in\! \C, ~\forall t \!\geq\! 0}$. The closed-loop system \eqref{eq:clsystem1} is {\em safe} on the \textit{safe set} $\C$ if $\C$ is forward invariant \cite{ames2017control}.

\begin{definition}[IRD] 
The \textit{input relative degree} of a sufficiently differentiable output function ${h \!:\! X \!\to\! \mathbb{R}}$ of system \eqref{system} on a set ${\mathcal{S} \!\subseteq\! X}$ with respect to ${u}$ is defined as an integer ${r \!\leq\! n}$ if ${ \forall x \!\in\! \mathcal{S}}$, ${L_{{{g}}}  L_{{{f}}}^{r-1} h(x)  \!\neq\! 0}$ and ${L_{{{g}}}  L_{{{f}}}^{i-1} h(x) \!=\! 0}$, for ${i \!\in\! \{1, 2, \ldots, r-1\}}$, where the higher-order Lie derivatives are defined recursively as 
\begin{equation*}
L_f^{i} h(x) \!\triangleq\! \grad L_f^{i-1} h(x) f(x) , \ 
L_g L_f^{i-1} h(x) \!\triangleq\! \grad L_f^{i-1} h(x) g(x)  ,
\end{equation*}
where ${\grad L_f^{i-1} h(x) \!\triangleq\! \derp{L_f^{i-1} h(x)}{x} }$, ${L_f^{0} h(x) \!\!\triangleq\! h(x)}$, with ${L_f^{i} h \!:\! X \!\!\to\!\! \mathbb{R}}$, ${\grad L_f^{i-1} h \!:\! X \!\!\to\!\! \mathbb{R}^{1 \times n}}$, and ${L_g L_f^{i-1} h \!:\! X \!\!\to\!\! \mathbb{R}^{m}}$.
\end{definition}

\begin{definition}[CBF, \cite{ames2017control}]
\label{def:cbf}
Let ${\C \!\subset\! X }$ be the 0-superlevel set of a continuously differentiable function ${h\!:\! X \!\to\! \mathbb{R}}$ that has IRD 1. Function $h$ is a \textit{control barrier function} for system \eqref{system} on $\C$ if ${\grad h(x) \!\triangleq\! \frac{\partial h (x)}{\partial x}  }$, satisfies ${\grad h(x) \!\neq\! 0}$ for all ${ x \!\in\! \partial \C}$ and there exists an extended class-$\mathcal{K}_{\infty}$ function\footnote{ 
A continuous function ${\alpha \!:\! [0, a ) \!\to\! \mathbb{R}^+}$, where ${a \!>\! 0}$, belongs to class-${\mathcal{K}}$ (${\alpha \!\in\! \mathcal{K}}$) if it is strictly monotonically increasing and ${\alpha(0) \!=\! 0}$. And, ${\alpha}$ belongs to class-${\mathcal{K}_\infty}$ (${\alpha \!\in\! \mathcal{K}_\infty}$) if ${a \!=\! \infty}$ and ${\lim_{r \!\to\! \infty} \alpha(r) \!=\! \infty}$. A continuous function ${\alpha \!:\! \mathbb{R} \!\to\! \mathbb{R}}$ belongs to the set of extended class-$\mathcal{K}_\infty$ functions (${\alpha \!\in\! \mathcal{K}_{\infty, e}}$) if it is strictly monotonically increasing, ${\alpha(0) \!=\! 0}$, ${\lim_{r \!\to\! \infty} \alpha(r) \!=\! \infty}$ and ${\lim_{r \!\to\! -\infty} \alpha(r) \!=\! -\infty}$. A continuous function ${\beta \!:\! [0, a ) \!\times\! \mathbb{R}^+_0 \!\to\! \mathbb{R}^+_0}$ belongs to class-${\mathcal{K L}}$ (${\beta \!\in\! \mathcal{KL}}$), if for every ${s \!\in\! \mathbb{R}^+_0}$, ${\beta(\cdot, s)}$ is a class-$\mathcal{K}$ function and for every ${r \!\in\! [0, a )}$, ${\beta(r, \cdot) }$ is decreasing and ${\lim_{s \!\to\! \infty} \beta(r, s) \!=\! 0.}$} ${\alpha \!\in\! \mathcal{K}_{\infty, e}}$ such that for all ${x \!\in\! \C}$:
\begin{equation}
\label{cbf}
   \sup_{u \in U} \big [ {\dot{h}(x, u)}  \big ] 
   \!=\! \sup_{u \in U} \big [ L_f h(x) + L_g h(x) u  \big ] 
   \!\geq\! -\alpha (h(x)).
\end{equation} 
\end{definition}
Then, given a CBF, formal safety guarantees can be established with the following theorem, based on Definition~\ref{def:cbf}:
\begin{theorem}
    \label{teo:cbfdef}
    If $h$ is a CBF for \eqref{system} on $\C$ with an ${\alpha \!\in\! \mathcal{K}_{\infty, e}}$, then any Lipschitz continuous controller ${\mathbf{k}\!:\! X \!\to\! U}$ satisfying 
    \begin{equation}
        \label{eq:cbf_def}
        \dot{h}\left(x, \mathbf{k}(x)\right) \geq - \alpha (h(x)),~~\forall x \in \C ,
    \end{equation}
    renders \eqref{eq:clsystem1} safe with respect to $\C$. 
\end{theorem}

Given a (possibly unsafe) locally Lipschitz continuous \textit{nominal controller} ${\mathbf{k_d} \!:\! X \!\to\! U}$, and a CBF ${h}$ with related ${\alpha}$ for system \eqref{system}, safety is ensured by solving the quadratic program (CBF-QP) \cite{ames2017control}:   
\begin{align*}
\begin{array}{l}
{\mathbf{k^*}(x)= \ }
\displaystyle  \argmin_{{u}  \in U} \ \ \ {\|u-\mathbf{k_d}(x) \|^2}  \\ [1mm]
~~~~~~~~~~~~\textrm{s.t.} ~~~~~~~~~\dot{h}(x, u)  \geq - \alpha (h(x)).
\end{array}
\end{align*}

Note that when ${L_g h (x) \!\equiv\! 0}$, control inputs do not appear in CBF condition \eqref{cbf}. Because the CBF-QP constraint does not depend on the variable ${u}$, we cannot synthesize a safe controller with this QP. This phenomenon is associated with the IRD, since condition \eqref{cbf} requires ${h}$ to have an IRD of one. In some applications, safety may require differentiation of ${h}$ with respect to system \eqref{system} until the control ${u}$ appears. In such cases, we can construct a higher-order CBF \cite{xiao2019control}.

For a differentiable function ${h \!:\! X \!\to\! \mathbb{R}}$, we consider a sequence of functions ${\phi_i \!:\! X \!\to\! \mathbb{R}}$ for ${i \!\in\! \{1, \ldots, r\!-\!1\}}$:
\begin{equation}
\label{eq:squphi}
    \phi_i (x) \!\triangleq\! \dot{\phi}_{i-1} (x) \!+ \alpha_i ( {\phi}_{i-1} (x) ), \ \phi_0 (x) \!\triangleq\! h(x),  
\end{equation}
where ${\alpha_i \!\in\! \mathcal{K}_{\infty, e}}$ is a ${(r\!-\!i)^{th}}$ order differentiable function. The associated extended sets and their intersections are defined as
\begin{equation}
\label{eq:hocbfset}
    \C_i \!\triangleq\! \{ x \!\in\! X \!:\! {\phi}_{i} (x) \!\geq\! 0 \}  , \
    \C \!\triangleq\! {\textstyle \bigcap\limits_{i = 0}^{r -1}} \C_i ,
\end{equation}
with ${\partial \C_i \!\triangleq\! \{ x \!\in\! X \!:\! {\phi}_{i} (x) \!=\! 0 \}}$. Safety with respect to set $\C$ is guaranteed via an HOCBF by assuming that $h$ has IRD $r$ on $\C$.

\begin{definition}[HOCBF, \cite{xiao2021high}, \cite{tan2021high}]
Let $\C$ be defined by \eqref{eq:hocbfset}. Then, $h$ is a \textit{high-order control barrier function} for system \eqref{system} on ${\C}$ if ${\frac{\partial {\phi}_{i}(x)}{\partial x} \!\neq\! 0}$ for all ${x \!\in\! \partial \C_i}$ and there exists ${\alpha_r \!\in\! \mathcal{K}_{\infty, e}}$ such that ${\forall x \!\in\! {\C }}$:
\begin{equation*}
   \!\!\!\sup_{u \in U} \! \Big [L_f^r h(x) \!+\! L_g L_f^{r-1} h(x) u \!+\! \mathcal{O}(x)\! \Big ] \!  \!\geq\! -\alpha_r ({\phi}_{r-1}(x)) ,
\end{equation*}
where ${\mathcal{O}(x) \!\triangleq\! \sum_{i= 1}^{r-1} L_f^i ( \alpha_{r-i} \circ  {\phi}_{r-i-1} )(x)}$, with ${\mathcal{O} \!:\! X \!\to\! \mathbb{R}}$.
\end{definition}
\begin{theorem}
\label{teo:hocbfdef}
If $h$ is an HOCBF for \eqref{system} on $\C$ defined in \eqref{eq:hocbfset}, then any Lipschitz continuous controller ${\mathbf{k}\!:\! X \!\to\! U}$ satisfying 
\begin{equation}
\label{eq:hocbf_def}
L_f^r h(x) \!+\! L_g L_f^{r-1} h(x) \mathbf{k}(x) \!+\! 
\mathcal{O}(x)  \geq -\alpha_r ({\phi}_{r-1}(x)) ,
\end{equation}
for all ${x \!\in\! \C}$ renders \eqref{eq:clsystem1} safe with respect to $\C$. 
\end{theorem}

Finally, using the HOCBF condition, an HOCBF-QP safety filter incorporating \eqref{eq:hocbf_def} can be constructed similarly to the CBF-QP.

\subsection{Analysis of the Matched and Unmatched Uncertainties} 
\label{sec:uncer}
In practice, uncertainties and disturbances cannot be fully modeled. Thus, functions $f$ and $g$ in the true model \eqref{system} may be imprecisely known, leading to degraded system performance and stability. In safe control synthesis, we typically use a control affine \textit{nominal model} to represent our best understanding of the system:  
\begin{equation*}
    \dot{x}  = \hat{f}(x) + \hat{g}(x) u , 
\end{equation*}
where ${\hat{f} \!:\! X \!\to\! \mathbb{R}^n}$, ${\hat{g}\!:\! X \!\to\! \mathbb{R}^{n \times m} }$ are locally Lipschitz continuous on $X$. That is, known functions $\hat{f}$ and $\hat{g}$ are used for controller design, but may differ from the actual model. To properly address the safe control design problem via uncertainty estimation, we assume:
\begin{assumption}
\label{as:gx}
The matrix $\hat{g}$ has full column rank, i.e., for all ${x \!\in\! X}$: ${\rank({\hat{g} (x))} \!=\! m}$, and the left pseudo-inverse of $\hat{g}$: ${\hat{g}^\dagger\!(x) \!=\! \big ( \hat{g}^\top\!(x) \hat{g} (x) \big ) ^{-1} \hat{g}^\top (x)}$, with ${\hat{g}^\dagger \!:\! X \!\to\! \mathbb{R}^{m \times n}}$, exists. 
\end{assumption}
\begin{remark}
    \label{re:uncer}
Most real-world input-affine robotic systems align with Assumption~\ref{as:gx}, which is a common assumption in uncertainty cancellation or disturbance rejection-based robust control \cite{xie2021disturbance}, \cite{sinha2022adaptive}. 
\end{remark}
To represent the discrepancies between the actual and nominal models, we assume that the actual model \eqref{system} can be equivalently described with the uncertain system model:
\begin{equation}
\label{sysun}
    \dot{x}  = \hat{f}(x) + \hat{g}(x) u + \Delta(t,\! x,\! u),
\end{equation}
where ${\Delta\!:\! \mathbb{R}^{+}_0 \!\times\! X \!\times\! U  \!\!\to\! \mathbb{R}^{n} }$ captures modeling inaccuracies and disturbance inputs. Observe that the CBF time derivative depends on ${\Delta}$, whose presence may violate \eqref{cbf}, causing unsafe control:
\begin{equation}
    \label{eq:hdot}
    {\dot{h}(t,\! x,\! u)}  =
   L_{\hat{f}} h(x) \!+\! L_{\hat{g}} h(x) u    
   \!+\!  \grad h(x)  \Delta(t,\! x,\! u).
\end{equation}

If we focus only on state and input-dependent model errors, a specific form of ${\Delta}$ can be found by adding and subtracting the nominal model from \eqref{system}:
\begin{equation}
\label{eq:sysune1}
    \dot{x}  = \hat{f}(x) + \hat{g}(x) u + \overbrace{ {f}(x) - \hat{f}(x)  + \ ( {g}(x) - \hat{g}(x) )  u}^{= \Delta(x, u) }  ,
\end{equation}
where ${{f}(x) \!-\! \hat{f}(x) \!\!\triangleq\! \Delta f(x)}$, ${{g}(x) \!-\! \hat{g}(x) \!\!\triangleq\! \Delta g(x)}$ 
are the unmodelled parts of ${\Delta }$.
Note that model description \eqref{eq:sysune1} implicitly ignores time-varying system disturbances, which implies the system \eqref{system} is time-invariant. Purely time-varying input disturbances in \eqref{sysun} of the form: ${\dot{x}  \!=\! \hat{f}(x) \!+\! \hat{g}(x) u \!+\! d(t)}$, with ${d \!:\! \mathbb{R}^{+}_0 \!\to\! \mathbb{R}^{n}}$, are assumed locally Lipschitz continuous in $t$ over ${t \geq 0}$. 
\begin{remark}
\label{re:uncdist}
In this study, unmodelled dynamics are considered a \textit{disturbance} if they depend only on time, e.g., $d(t)$. Otherwise, they are considered an \textit{uncertainty} when they do not explicitly depend on time, but depend on states and control inputs. The compound modeling inaccuracy ${\Delta(t,\! x,\! u)}$ given in \eqref{sysun} is deemed an \textit{uncertainty}.
\end{remark}

Assumption~\ref{as1} ensures well-posedness of our uncertainty estimator. 
\begin{assumption}
\label{as1}
The uncertainty $\Delta$ in \eqref{sysun} and its time derivative are bounded by some ${\delta_b, \delta_L \!\in\! \mathbb{R}^+}$ for all ${(t,\!x,\!u) \!\in\! \mathbb{R}^{+}_0 \!\times\! X \!\times\! U }$: ${\left\| \Delta(t,\! x,\! u) \right\|  \!\leq\! \delta_b}$, ${
{ \big\| \dot{\Delta}(t,\! x,\! u) \big\|  }  \!\leq\! \delta_L }$.
\end{assumption}
\begin{remark}
\label{re:uncer_exp} 
Assumption~\ref{as1} is reasonable for many practical nonlinear systems, such as wheeled mobile robots \cite{xie2021disturbance}, aerial robots \cite{yan2023surviving}, and ships \cite{du2016robust}. The upper bounds ${\delta_b, \delta_L }$ can be learned from data obtained from system simulations or tests. Moreover, note that ${\delta_b}$ need not be known for our proposed CBF-based safety guarantees, but is needed to characterize the estimation error dynamics. 
\end{remark}

\begin{definition}[Full-State Matched/Unmatched Uncertainty]
For all ${(t,\!x,\!u) \!\in\! \mathbb{R}^{+}_0 \!\times\! X \!\times\! U }$, uncertainty ${\Delta}$ in \eqref{sysun} is the sum of a \textit{full-state matched uncertainty} ${\Delta_m\!:\! \mathbb{R}^{+}_0 \!\times\! X \!\times\! \mathbb{R}^{m}  \!\to\! \mathbb{R}^{n}}$, and a \textit{full-state unmatched uncertainty}, ${\Delta_u\!:\! \mathbb{R}^{+}_0 \!\times\! X \!\times\! \mathbb{R}^{m}  \!\to\! \mathbb{R}^{n}}$: 
\begin{align}
    \label{eq:munmdel}
    \Delta(t,\! x,\! u) = \Delta_m(t,\! x,\! u) + \Delta_u(t,\! x,\! u) ,
\end{align}
where, uncertainty ${\Delta_m(t,\! x,\! u)}$ lies in the image of ${\hat{g}}$: ${\Delta_m(t,\! x,\! u) \!\in\! \Ima(\hat{g}(x))}$. Conversely, uncertainty ${\Delta_u \!\in\! \Delta }$ does not lie in the image of ${\hat{g}}$: ${\Delta_u(t,\! x,\! u) \!\perp\! \Ima(\hat{g}(x))}$.
\label{def:mun1}
\end{definition} 
Since Definition~\ref{def:mun1} implies that $\Delta_m$ can be described as a linear combination of the columns of $\hat{g}$, $\Delta$ can be expressed as
\begin{align}
\label{eq:matched_der}
\Delta(t,\! x,\! u) = \hat{g}(x) \vartheta(t,\! x,\! u) + \Delta_u(t,\! x,\! u), 
\end{align}
where ${\vartheta\!:\! \mathbb{R}^{+}_0 \!\times\! X \!\times\! U  \!\to\! \mathbb{R}^{m}}$ represents the full-state matched part of $\Delta$ in the control input channel of the system. Multiplying both sides of \eqref{eq:matched_der} with $\hat{g}^\top\!(x)$ yields:
\begin{equation*}
\hat{g}^\top\!(x) \Delta(t,\! x,\! u) \!=\! \hat{g}^\top\!(x) \hat{g}(x) \vartheta(t,\! x,\! u) \!+\! \hat{g}^\top\!(x) \Delta_u(t,\! x,\! u),
\end{equation*}
where ${\hat{g}^\top\!(x) \Delta_u(t,\! x,\! u) \!\equiv\! 0 }$ since ${\Delta_u(t,\! x,\! u) \!\perp\! \Ima(\hat{g}(x))}$. 
By Assumption~\ref{as:gx}, ${\hat{g}^\dagger\!(x)}$ exists for all ${x \!\in\! X}$. Therefore, multiplying both sides of this equation by ${( \hat{g}^\top\!(x) \hat{g}(x)  )^{-1}}$ yields: ${\vartheta(t,\! x,\! u)  \!=\! \hat{g}^\dagger\!(x)  \Delta(t,\! x,\! u)}$. Then, the full-state matched part of the augmented uncertainty is obtained by substituting $\vartheta$ into \eqref{eq:matched_der}:
\begin{align}
\label{eq:matched_fin}
\Delta_m(t,\! x,\! u) = \hat{g}(x) \hat{g}^\dagger\!(x)  \Delta(t,\! x,\! u). 
\end{align}
Subtracting \eqref{eq:matched_fin} from \eqref{eq:munmdel} gives the full-state unmatched uncertainty ${\Delta_u}$, which is the projection of ${\Delta}$ onto the null space of $\hat{g}$:
\begin{align}
\label{eq:unmatched_fin}
\Delta_u(t,\! x,\! u) = \big (\mathbf{I} - \hat{g}(x) \hat{g}^\dagger\!(x) \big ) \Delta(t,\! x,\! u). 
\end{align}

Equations \eqref{eq:matched_der} and \eqref{eq:matched_fin} imply that the full-state matched uncertainty influences the system via the control input channel and the uncertain system model \eqref{sysun} is equivalent to  
\begin{equation*}
    \dot{x}  \!=\! \hat{f}(x) + \hat{g}(x)  \big (u + \hat{g}^\dagger\!(x) \Delta(t,\! x,\! u) \big) +  \Delta_u(t,\! x,\! u) .
\end{equation*}
Hence, the time derivative of $h$ in \eqref{eq:hdot} can be rewritten as
\begin{align}
\begin{split}
    \label{eq:equva}
    {\dot{h}(t,\! x,\! u)}  =
   L_{\hat{f}} h(x) + L_{\hat{g}} h(x) \big (u + \hat{g}^\dagger\!(x)  \Delta(t,\! x,\! u) \big )  \\ 
   +  \grad h(x) \big (\mathbf{I} - \hat{g}(x) \hat{g}^\dagger\!(x) \big ) \Delta(t,\! x,\! u) .
\end{split}
\end{align}

Our study uses Definition~\ref{def:mun1} for uncertainty compensation via an uncertainty estimator. However, this definition does not consider a specific state-dependent output function, $h(x)$. The following definitions describe matching conditions that are relevant to a scalar output function. First, like the IRD, we define the \textit{disturbance relative degree} (DRD) for the scalar function ${h}$ with respect to model \eqref{sysun}:
\begin{definition}[DRD]
The \textit{disturbance relative degree} of a sufficiently differentiable output function ${h(x), h\!:\! X \!\to\! \mathbb{R},}$ of the uncertain system \eqref{sysun} on a set ${\mathcal{S} \!\subseteq\! X}$ with respect to $\Delta$, is defined as an integer ${v \!\leq\! n}$ if ${L_{\Delta}  L_{\hat{f}}^{v-1} h(t, x, u)  \!\neq\! 0}$, ${\forall (t, x,u) \!\in\! \mathbb{R}^{+}_0 \!\times\! \mathcal{S} \!\times\! U }$, and ${L_{\Delta}  L_{\hat{f}}^{i-1} h(t, x, u) \!\equiv\! 0}$, where ${L_{\Delta} L_{\hat{f}}^{i-1} h(t, x, u) \triangleq \grad L_{\hat{f}}^{i-1} h(x) \Delta}$, for ${i \!\in\! \{1, \ldots, v-1\}}$. 
\end{definition}

\begin{definition}[Output-Matched (-Unmatched) Uncertainty]
The uncertainty ${\Delta}$ is an \textit{output-matched (-unmatched) uncertainty} for system \eqref{sysun} with respect to a sufficiently differentiable output function ${h \!:\! X \!\to\! \mathbb{R}}$ if ${\text{DRD} = \text{IRD}}$ (${\text{DRD} \neq \text{IRD}}$).
\label{def:mun}
\end{definition} 

\begin{remark}
\label{re:unmatchedre}
Definition~\ref{def:mun1} implies a sufficient (but not necessary) condition for output-matched uncertainty:  $\left ( \Delta_m(t,\! x,\! u) \!\in\! \Ima(\hat{g}(x)) \right ) \!\!\implies\!\! \left( \text{DRD} \!=\! \text{IRD} \right )$.  It is stronger than the condition in Definition~\ref{def:mun}.
The second condition in Definition~\ref{def:mun} is a sufficient condition for full-state unmatched uncertainty: ${\left ( \text{DRD} \!\neq\! \text{IRD} \right )  \!\implies\!\! \left ( \Delta_u(t, \!x, \!u) \!\perp\! \Ima(\hat{g}(x)) \right )}$ for some ${(t,\! x,\!u) \!\in\! \mathbb{R}^{+}_0 \!\times\! X \!\times\! U }$. Practically speaking, this condition implies that even if full-state uncertainty elimination is not possible via feedback, one may be able to eliminate the uncertainty from specific output functions, such as CBF $h$.
\end{remark}

Our study utilizes Definition~\ref{def:mun} in the CBF and HOCBF conditions as they apply to an output function $h$. And, to reduce the complexity of HOCBFs in the presence of output unmatched uncertainty, we assume that the DRD of system \eqref{sysun} is less than IRD by at most one: 
\begin{assumption} 
\label{as:DRD-IRD}
The compound uncertainty $\Delta$ in system \eqref{sysun} satisfies ${\text{DRD} \!\geq\! r \!-\! 1}$ when ${\text{IRD} \!=\! r}$ for ${r \!\geq\! 2}$  with respect to a sufficiently differentiable output function ${h \!:\! X \!\to\! \mathbb{R}}$ of system \eqref{sysun}.
\end{assumption}

Under Assumption~\ref{as:DRD-IRD}, suppose ${\text{IRD} \!=\! r}$ and ${ \text{DRD} \!=\! r \!-\!1}$, and ${r \!\geq\! 2}$, then the higher-order time derivatives of $h$ with respect to the system dynamics \eqref{sysun} are given by ${h^{i}(x)  \!=\!  L_{\hat{f}}^{i} h(x),   i \!\in\! \{1, 2, \ldots, r \!-\! 2\}}$, and
\begin{eqnarray}
\label{eq:hocbfder1}
&{h}^{r-\!1}(t,\! x,\! u) \!=\! L_{\hat{f}}^{r-\!1} h(x) \!+\! L_{\Delta} L_{\hat{f}}^{r-2} h(t,\! x,\! u) , \nonumber \\
&{h}^{r}(t,\! x,\! u) \!=\! L_{\hat{f}}^{r} h(x) \!+\! L_{\hat{g}} L_{\hat{f}}^{r-1} h(x) u \!+\!  L_{\Delta} L_{\hat{f}}^{r-1} h(t,\! x,\! u) \nonumber \\ 
&+ \dfrac{d}{dt} \! \big (\! \grad L_{\hat{f}}^{r-2} h(x)  \! \big ) \! {\Delta}(t,\! x,\! u) \!+\! \grad L_{\hat{f}}^{r-2} h(x)   \Dot{\Delta}(t,\! x,\! u)  .
\end{eqnarray}
One can observe from \eqref{eq:hocbfder1} that, if ${\text{IRD} \!=\! \text{DRD} \!=\! r}$, then
\begin{equation*}
{h}^{r}(t,\! x,\! u) = L_{\hat{f}}^{r} h(x) \!+\! L_{\hat{g}} L_{\hat{f}}^{r-1} h(x) u  \!+\!  L_{\Delta} L_{\hat{f}}^{r-1} h(t,\! x,\! u)  ,
\end{equation*}
which shows that the uncertainty $\Delta$ only affects the highest-order derivative of $h$. When ${(\text{IRD} \!=\! r ) < \text{DRD}}$, the term $\Delta$ disappears from \eqref{eq:hocbfder1}: it does not affect the HOCBF conditions.

Ignoring the effects of $\Delta$ in \eqref{eq:equva} or \eqref{eq:hocbfder1} may lead to unsafe behavior. To remedy this problem, we estimate and bound the uncertainty in a quantifiable way, compensate the full-state matched uncertainty, and incorporate the error bound and unmatched part into the CBF and HOCBF conditions. The problem statement of this work follows:
\begin{problem}
\label{pro:first}
Given the nominal system dynamics ${ \dot{x}  \!=\! \hat{f}(x) \!+\! \hat{g}(x) u}$ in \eqref{sysun} and the upper bounds ${\delta_L, \delta_b \!\in\! \mathbb{R}^{+}_0}$ for ${\Delta }$, which can be a composition of output-matched and output-unmatched uncertainties with respect to CBF ${h}$, and a safe set $\C$ defined by \eqref{eq:CBF1}, synthesize a safe controller under Assumption~\ref{as:DRD-IRD}. 
\end{problem}

\section{Main Result}
\label{sec:main}
This section first proposes and analyzes an uncertainty estimator. Then, we introduce our estimator-based robust CBF frameworks. 

\subsection{Uncertainty Estimator}
\label{sec:ue}
To estimate the uncertainty ${\Delta}$ in \eqref{sysun}, we propose an uncertainty estimator with the following structure:
\begin{align}
    \label{eq:bhat}
    \hat{\Delta}(x,\xi) &= {\Lambda} x - \xi, \\
    \dot{\xi} &= \Lambda \big(\hat{f}(x) + \hat{g}(x) u + \hat{\Delta}(x,\xi) + s_n \big), 
    \label{eq:xidot}
\end{align}
where ${\hat{\Delta} \!:\! X \!\times\! \mathbb{R}^{n} \!\to\! \mathbb{R}^{n}}$ is the estimated uncertainty, ${\xi \!\in\! \mathbb{R}^{n}}$ is an auxiliary state vector, ${0 \!\prec\! \Lambda \!\in\! \mathbb{R}^{n \times n}}$ is a positive definite  design matrix, and ${s_n \!\in\! \mathbb{R}^{n}}$ is bounded measurement noise: ${\|s_n(t)\|_{\infty} \!\triangleq\! \sup_{t} {\|s_n(t)\|} \!\leq\! \bar{s}_n}$ for some known ${\bar{s}_n \!\in\! \mathbb{R}^+_0}$. Without loss of generality, we set the initial value of ${\hat{\Delta}}$, ${\hat{\Delta}(0) \!=\! {\bf 0}}$, by assigning ${\xi(0) \!=\! \Lambda x(0)}$ in \eqref{eq:bhat}. The uncertainty estimation error is 
\begin{align}
\label{eq:errdot}
    e(t,\! x,\! u) = {\Delta}(t,\! x,\! u) - \hat{\Delta}(t) , ~ e_0 \triangleq e(t \!=\! 0) \in \mathbb{R}^{n}.
\end{align}
By slight abuse of notation, we denote ${\hat{\Delta}(x,\xi})$ briefly as ${\hat{\Delta}(t)}$.

The uncertainty estimator form in \eqref{eq:bhat}, \eqref{eq:xidot} implicitly assumes known control inputs and full state measurement/estimation. This estimator extends the disturbance observer studied in \cite{dacs2022robust, alan2022disturbance,stotsky2002} to nonlinear systems with multiple disturbance inputs. The following Lemma characterizes the uniformly ultimately boundedness property of the estimation error dynamics. 
\begin{lemma}
Consider the uncertain system \eqref{sysun} with a compound uncertainty function ${\Delta \!:\! \mathbb{R}^{+}_0 \!\times\! X \!\times\! U  \!\to\! \mathbb{R}^{n} }$ that satisfies Assumption~\ref{as1} with upper bounds ${\delta_L, \delta_b \!\in\! \mathbb{R}^{+}_0}$, and the uncertainty estimator \eqref{eq:bhat}-\eqref{eq:xidot} with a positive definite matrix ${\Lambda \!\in\! \mathbb{R}^{n \times n}}$. Then, error $e$ given in \eqref{eq:errdot} is uniformly ultimately bounded.
\label{lem:Velyap}
\end{lemma}
\begin{proof}
From Equations \eqref{sysun}, \eqref{eq:bhat}, \eqref{eq:xidot} we have
\begin{equation}
    \dot{e}(t,\! x,\! u) = \dot{{\Delta}}(t,\! x,\! u) - \Lambda e(t,\! x,\! u) + \Lambda s_n.
    \label{eq:dote}
\end{equation}
By omitting the dependence on (${t,\! x,\! u}$) for simplicity, we consider a candidate Lyapunov function: ${V_e \left (e \right) \!=\! \frac{1}{2} e^\top e}$, with ${V_e \!:\! \mathbb{R}^{n} \!\to\!  \mathbb{R}^+_0}$. The time derivative of ${V_e}$ along the trajectory of \eqref{eq:dote} satisfies: 
\begin{align}
\begin{split}
    \dot{V}_e &= e^\top ( \dot{{\Delta}} - \Lambda e + \Lambda s_n  ) 
     \leq - e^\top \Lambda e + e^\top \Lambda s_n  +  \| e  \| \| \dot{{\Delta}}  \| \\
    & \leq - e^\top \Lambda e + e^\top \Lambda s_n +  \| e  \| \delta_L \\
     & \leq - \lambda_{\rm min} (\Lambda)  \| e   \| ^2 + \lambda_{\rm max} (\Lambda)  \| e   \| \bar{s}_n +  \| e  \| \delta_L. 
     \label{eq:dotV_intermediate}  
 \end{split}
\end{align}
The last inequality is Rayleigh's inequality for a real symmetric positive definite matrix $\Lambda$: ${\lambda_{\rm min}(\Lambda) \| e \| ^2 \!\leq\! e^\top \Lambda e \!\leq\! \lambda_{\rm max}(\Lambda) \|e\|^2}$, where $\lambda_{\rm min}(\Lambda)$ is the smallest eigenvalue of $\Lambda$, and $\lambda_{\rm max}(\Lambda)$ is the largest eigenvalue of $\Lambda$. To replace ${\| e \| ( \delta_L \!+ \lambda_{\rm max} (\Lambda)  \bar{s}_n)}$ with an upper bound in \eqref{eq:dotV_intermediate}, we introduce an inequality, ${\big ( \| e \| \sqrt{\lambda_{\rm min}(\Lambda)}  \!-\!  {( \delta_L \!+\! \lambda_{\rm max} (\Lambda)  \bar{s}_n)}/\!{\sqrt{\lambda_{\rm min}(\Lambda)}} \big )^2 \!\geq\! 0 }$, which yields: ${ \| e \| ( \delta_L \!+\! \lambda_{\rm max} (\Lambda)  \bar{s}_n) \!\leq\! \frac{1}{2} {\lambda_{\rm min}(\Lambda)} \| e \| ^2 \!+\! \frac{( \delta_L \!+\! \lambda_{\rm max} (\Lambda)  \bar{s}_n)^2}{2 \lambda_{\rm min}(\Lambda)} }$. Substituting this inequality into \eqref{eq:dotV_intermediate} yields:
\begin{equation}
\dot{V}_e   \leq - \frac{{\lambda_{\rm min}(\Lambda)} }{2} \| e \| ^2 + \frac{( \delta_L + \lambda_{\rm max} (\Lambda)  \bar{s}_n)^2}{2 \lambda_{\rm min}(\Lambda)}.
 \label{eq:er_iss}
\end{equation}

Under Assumption~\ref{as1}, and the initialization ${\hat{\Delta}(0) \!=\! {\bf 0}}$, the initial estimation error will be bounded by ${\| e_0 \| \!\leq\! \delta_b}$. The solution to ${e}$ in \eqref{eq:dote} with state transition matrix ${ {\rm e}^{-\Lambda t}}$ is given by
\begin{equation}
\label{eq:transition}
    {e}(t) =  {\rm e}^{-\Lambda t} e_0 + \int_{0}^t  {\rm e}^{-\Lambda (t-\tau)} ( \dot{{\Delta}} + \Lambda s_n)  (\tau) d \tau .
\end{equation}
Taking the norms of both sides of \eqref{eq:transition} and using the triangle inequality, then integrating the right-hand side yields ${\forall t \!\in\! \mathbb{R}^{+}_0}$:
\begin{equation*}
     \| {e}(t) \| \leq \big \|{\rm e}^{-\Lambda t} \big \|  \underbrace{\|  e_0 \|}_{\leq \delta_b} + \int_{0}^t \big \| {\rm e}^{-\Lambda (t-\tau)} \big \| \underbrace{\| \dot{{\Delta}}(\tau) + \Lambda s_n \|}_{\leq \delta_L \!+ \lambda_{\rm max} (\Lambda)  \bar{s}_n } d \tau  , 
\end{equation*}
To bound $\big \|{\rm e}^{{-\Lambda} t} \big \|$, consider a Lyapunov equation: $ { (-\Lambda)}^\top P \!+\! P {(-\Lambda)}  + H = 0 $, where $P, H \!\in\! \mathbb{R}^{n \times n}$, $H \!=\! H^\top \!\succ\! 0$, and $P \!=\! P^\top \!\succ\! 0$ is the unique equation solution. Then, we have $\big \|{\rm e}^{{-\Lambda} t} \big\| \leq D  {\rm e}^{ -\tau_e t }$, where $D \!\triangleq\!  \sqrt{\| P \| \| P^{-1}\| }\!$~, $ \tau_e \!\triangleq\! \lambda_{\rm min} (H) /  ( {2 \| P \|}   )$. Then, we obtain
\begin{equation*}
     \| e(t) \| \!\!\leq\! D \! \Big(\!\delta_b \!- \frac{\delta_L \!\!+\! \lambda_{\rm max} (\Lambda) \bar{s}_n}{\tau_e} \! \Big) \text{e}^{-\tau_e t} \!+ \frac{D (\delta_L \!\!+\! \lambda_{\rm max} (\Lambda) \bar{s}_n \! )}{\tau_e}  \!\!\triangleq\! \bar{e}(t) . 
\end{equation*}
We can also simply take ${D = 1}$ and ${\tau_e = \lambda_{\rm min} (\Lambda)}$ if ${\Lambda}$ is symmetric positive definite.
\vskip -0.1 true in
\end{proof}
\vskip -0.1 true in
We use the estimated uncertainty to compensate for the matched uncertainty via a composite feedback law that minimizes discrepancies between the actual and nominal system models. To achieve this, we integrate the uncertainty estimator into the control:
\begin{equation}
    \label{eq:u_hat}
    \mathbf{{k}} (x, \hat{\Delta}) = \mathbf{\bar{k}}(x) - \hat{g}^\dagger\!(x) \hat{\Delta}  \triangleq \mathbf{\bar{k}}(x) -
     \mathbf{k}_{\hat{\Delta}} (x, \hat{\Delta}), 
\end{equation}
where ${ \bar{u} \!=\! \mathbf{\bar{k}}(x)}$ is the controller that is designed for the nominal system, ${u_{\hat{\Delta}} \!=\! \mathbf{k}_{\hat{\Delta}}(x, \hat{\Delta})}$, ${\mathbf{k}_{\hat{\Delta}} \!:\! X \!\times\! \mathbb{R}^{n} \!\to\! U}$ is the projection of the matched estimated uncertainty onto the system's input channel for uncertainty attenuation (see Fig.~\ref{fig:sector}), and ${u \!=\! \mathbf{{k}} (x, \hat{\Delta})}$, ${\mathbf{{k}} \!:\! X \!\times\! \mathbb{R}^{n} \!\to\! U}$ is the \textit{composite control law}. Since ${u_{\hat{\Delta}}}$ lies in the image of ${\hat{g}}$, it can be compensated via ${\mathbf{\bar{k}}}$. Substituting control law \eqref{eq:u_hat} into the uncertain system dynamics \eqref{sysun} results in the closed-loop dynamics:
\begin{equation}
\label{syscon}
     \dot{x}  =  \hat{f}(x) + \hat{g}(x) \mathbf{\bar{k}}(x) + \hat{g}(x) \hat{g}^\dagger(x) e(t, x, u) + \Delta_u(t, x, u) ,
\end{equation} 
which depends on the uncertainty estimation error ${e}$ and the full-state unmatched uncertainty ${\Delta_u}$.

\subsection{Uncertainty Estimation-Based Safety with CBFs}
To guarantee CBF robustness, we must consider the uncertainty ${\Delta}$ in our safety analysis, even though it cannot be measured in practice. To address this issue, we replace ${\Delta}$ with the estimated uncertainty ${\hat{\Delta}}$ and a lower bound on the estimation error ${{e}}$. Theorem~\ref{th:met3} summarizes the robust safety assurances provided by uncertainty estimation and full-state matched uncertainty cancellation. Note that here we consider CBFs with ${\text{IRD} \!=\! 1}$; therefore, from Assumption~\ref{as:DRD-IRD}, ${\text{DRD} \!\geq\! 1}$.  
\begin{theorem}
\label{th:met3}
Consider system \eqref{sysun} with uncertainty ${\Delta}$ that includes full-state matched and unmatched components and satisfies Assumption~\ref{as1}, an uncertainty estimator and a feedback control law ${{u} \!=\!  \bar{u} \!-\! \hat{g}^\dagger\!(x) \hat{\Delta}(t)  }$ given in \eqref{eq:bhat}, \eqref{eq:xidot}, and \eqref{eq:u_hat}, respectively. Let $h$ be a CBF for \eqref{sysun} on set ${\C \!\triangleq\! \left\{ x \!\in\! X  ~|~ h(x) \!\geq\! 0 \right\}}$ such that ${\grad h \!\neq\! 0}$ for all ${ x \!\in\! \partial \C}$. Any Lipschitz continuous controller ${\bar{u} \!=\! \mathbf{\bar{k}}(x)}$ satisfying
\begin{align} 
\begin{split}
\label{eq:thm4_set}
L_{\hat{f}} h(x) + L_{\hat{g}} h(x) ({\mathbf{\bar{k}}(x) - \hat{g}^\dagger\!(x) \hat{\Delta}(t)} ) + \grad h(x) \hat{\Delta}(t) \\
- \norm{\grad h(x)} \bar{e}(t) \geq -\alpha (h(x)),
\end{split}
\end{align}
${\forall x \!\in\! \C}$ renders set ${ \C }$ forward invariant: ${x(0) \!\in\! \C \!\!\! \implies\!\!\! x(t) \!\in\! \C, \forall t \!\geq\! 0}$.
\end{theorem}
\begin{proof} 
From \eqref{eq:u_hat} and \eqref{syscon}, the closed-loop system
dynamics is:
\begin{equation}
    \label{eq:proof_theo} 
    \Dot{x} =  \hat{f}(x) + \hat{g}(x) \big ({\mathbf{\bar{k}}(x) - \hat{g}^\dagger\!(x) \hat{\Delta}(t)} \big ) +  \hat{\Delta}(t) + e(t).
\end{equation}
According to Theorem~\ref{teo:cbfdef}, safety is ensured with respect to $\C$ by CBF condition \eqref{eq:cbf_def} assuming that ${\text{IRD} \!=\! 1}$. Substituting \eqref{eq:proof_theo} into \eqref{eq:cbf_def} yields:
\begin{align}
\begin{split}
    \label{pro:main_pro2}
    L_{\hat{f}} h(x) + L_{\hat{g}} h(x) \big ({\mathbf{\bar{k}}(x) - \hat{g}^\dagger\!(x) \hat{\Delta}(t)} \big ) + \grad h(x) \hat{\Delta}(t) \\
+ \grad h(x) e(t) \geq -\alpha (h(x)) , ~~\forall x \in \C,
\end{split}
\end{align}
where estimation error $e$ is unknown. Thus, ${\grad h(x) e(t)  }$ is also unknown. But, Lemma~\ref{lem:Velyap} upper bounds the estimation error. To obtain a sufficient CBF condition, we replace the unknown term in \eqref{pro:main_pro2} with the lower bound ${\left(- \norm{\grad h(x)} \bar{e}(t)   \right)}$: 
\begin{align}
    \label{pro:main_pro3}
      L_{\hat{f}} h(x) + L_{\hat{g}} h(x) \big({\mathbf{\bar{k}}( x) - \hat{g}^\dagger\!(x) \hat{\Delta}(t)} \big) + \grad h(x) \hat{\Delta}(t)  \nonumber \\
 - \norm{\grad h(x)} \bar{e}(t) \geq -\alpha (h(x)), ~~\forall x \in \C,  
\end{align}
which implies that ${\Dot{h}(t,\! x,\! u) \!\geq\! -\alpha (h(x)), \forall x \!\in\! \C}$, where ${\dot{h}}$ is given in \eqref{eq:equva}; therefore, ${x(t) \!\in\! \C, ~\forall t \!\geq\! 0}$ when ${x_0 \!\in\! \C}$.
\end{proof}
We remark that given a locally Lipschitz continuous nominal controller, $\mathbf{k_d}$, Theorem~\ref{th:met3} allows for the design of a safe controller by formulating an uncertainty estimation and cancellation-based robust safety filter, where the constraint is expressed by \eqref{pro:main_pro3}. This approach enables the implementation of the controller as an efficient QP.

Note that the robust CBF condition \eqref{pro:main_pro3} explicitly depends on the time-varying upper error bound, ${\bar{e}}$. A similar method was used in \cite{dacs2022robust} to address robustness against unmodeled system dynamics by estimating the effect of the disturbance on the CBF. While safety was guaranteed in \cite{dacs2022robust}, uncertainty compensation via estimation was not considered, and the proposed safety filter may perform poorly.

By Lemma~\ref{lem:Velyap}, the estimation error is bounded, which is characterized by the Lyapunov function ${V_e}$. Using this property, we incorporate the uncertainty estimator \eqref{eq:bhat}, \eqref{eq:xidot} into the CBF construction with the augmentation of the existing CBF. In particular, we modify the CBF $h$ to provide robustness against the estimation error $e$:
\begin{equation}
\label{hecbf}
    h_V (x, e)  \triangleq h(x) - \sigma_V V_e(e) ,  
\end{equation}
where ${\sigma_V \!\in\! \mathbb{R}^+}$. And, ${\C_V \!\triangleq\! \big \{ [x~e]^\top \!\in\! X \!\times\! \mathbb{R}^{n}  \!~|\!~ h_V(x, e) \!\geq\! 0 \big \}}$ is a subset of safe set ${\C \!\triangleq\! \big\{ x \!\in\! X  ~|~ h(x) \!\geq\! 0 \big\}}$ since ${V_e}$ is a Lyapunov function. We assume that ${\frac{\partial h_V}{\partial x_e} \!\neq\! 0}$, for all ${\partial {\C_V}}$, where ${x_e \!\triangleq\! [x~e]^\top}$. The following theorem uses the estimator dynamics to relate the controllers designed for nominal system safety to the safety of the uncertain system. Note that for Theorem \ref{theo:main} we use a specific extended class-${\mathcal{K}_{\infty, e}}$ function, ${\alpha (h(x)) = \alpha_h h(x)}$, where ${\alpha_h \!\in\! \mathbb{R}^+}$, for simplicity. The proof drops function arguments for convenience.
\begin{theorem}
\label{theo:main}
Under the same assumptions of Theorem~\ref{th:met3}, let $h_V$ be a CBF for \eqref{sysun} on ${\C_V}$. Consider the Lyapunov function of the uncertainty estimator dynamics with the inequality in \eqref{eq:er_iss}: ${\Dot{V}_e  \!\leq\! -2 \mu_e \| e \| ^2 \!+\! \gamma}$, with ${\mu_e \!\triangleq\! {\lambda_{\rm min}(\Lambda)}/{4} }$ and ${ \gamma \!\triangleq\! \frac{( \delta_L \!+ \lambda_{\rm max} (\Lambda)  \bar{s}_n)^2}{2 \lambda_{\rm min}(\Lambda)}}$, and a constant ${\alpha_h \!\in\! \mathbb{R}^+}$ such that ${\mathcal{E} \!\triangleq\! 2 \sigma_{V}\mu _e \!-\! {\sigma_V \alpha_h}/{2}  \!>\! 0}$. Then, any Lipschitz continuous controller ${\bar{u} \!=\! \mathbf{\bar{k}}(x)}$ satisfying
\begin{align}
\begin{split}
\label{heineq}
    L_{\hat{f}} h(x) + L_{\hat{g}} h(x) \big ( \mathbf{\bar{k}}(x) - \hat{g}^\dagger\!(x) \hat{\Delta}(t) \big )  + \grad h(x) \hat{\Delta}(t) \\ \geq  -\alpha_h h(x) 
 + \frac{1}{4 \mathcal{E} } \norm{\grad h(x)} ^2  + \sigma_V \gamma , 
\end{split}
\end{align}
for all ${x_e \!\in\! \C_V}$ renders the set ${ \C_V }$ forward invariant. 
\end{theorem}
\begin{proof}
Our goal is to prove that ${\dot{h}_V \!\geq\!  -\alpha_h h_V, ~ \forall t \!\geq\! 0}$, which implies that ${\dot{h} \!\geq\!  -\alpha_h h, ~ \forall t \!\geq\! 0}$. The time derivative of ${h_V}$ in \eqref{hecbf} satisfies:
\begin{align*}
\begin{split}
  \dot{h}_V &\!=\!  \dot{h} - \sigma_V \dot{V}_e  \\
  &\!\geq\! \dot{h} - \sigma_V (-2\mu_e \|e\|^2 + \gamma ) \\
    &\!=\! L_{\hat{f}} h \!+\! L_{\hat{g}} h ({\mathbf{\bar{k}}(x) \!-\! \hat{g}^\dagger \hat{\Delta}} )  \!+\! \grad h  \hat{\Delta} 
    \!+\! \grad h  e + 2 \sigma_V \mu_e \|e\|^2 \!-\! \sigma_V \gamma \\
     &\!=\!  L_{\hat{f}} h \!+\! L_{\hat{g}} h ({\mathbf{\bar{k}}(x) \!-\! \hat{g}^\dagger \hat{\Delta}} )  \!+\! \grad h  \hat{\Delta} 
     \!+\! \grad h  e
     - \sigma_V \gamma \\
     &~~+ \underbrace{\left ( 2 \sigma_V \mu_e -  {\sigma_V\alpha_h}/{2} \right )}_{= \mathcal{E}} \|e\|^2 + {\sigma_V\alpha_h} \|e\|^2 / {2}  
\end{split}
\end{align*}
\begin{align*}
\begin{split}     
     &\!=\! L_{\hat{f}} h \!+\! L_{\hat{g}} h ({\mathbf{\bar{k}}(x) \!-\! \hat{g}^\dagger \hat{\Delta}} )  \!+\! \grad h  \hat{\Delta}  - \sigma_V \gamma     \\
     &~~+ \Big ( \sqrt{\mathcal{E}}e +  
     \dfrac{1}{2 \sqrt{ \mathcal{E}} } {\grad h}  \Big )^2 - \dfrac{1}{4 \mathcal{E} } \norm{\grad h} ^2 + {\sigma_V\alpha_h} \|e\|^2 / {2}  \\
     &\!\!\!\!\!\geq\!\! \underbrace{L_{\hat{f}} h \!+\!\! L_{\hat{g}} h ({\mathbf{\bar{k}}(x) \!\!-\!\! \hat{g}^\dagger \hat{\Delta}} \!)  \!+\!\! \grad h  \hat{\Delta}  \! -\! \dfrac{1}{4 \mathcal{E} } \! \norm{\grad h} ^2 \!\!-\! \sigma_V \gamma }_{\geq -\alpha_h h} \!+  {\sigma_V \alpha_h \|e\|^2}\!/{2} \\
     &\implies \dot{h}_V \geq -\alpha_h h + {\sigma_V \alpha_h \|e\|^2}/{2} = -\alpha_h h_V
\end{split}
\end{align*}
This result guarantees that ${x_e(t) \!\in\! \C_V}$, ${\forall t \!\geq\! 0}$ if ${x_e(0) \!\in\! \C_V}$, which implies 
${h(x(t)) \!\geq\! 0}$, ${\forall t \!\geq\! 0}$. 

The second proof line uses the upper bound of ${V_e}$ found in \eqref{eq:er_iss}. Note that ${\dot{h}}$ is given in \eqref{pro:main_pro2}: it is used in the third proof line. The next line uses a square to define a new lower bound. Finally, we obtain the main statement using the condition in the theorem.
\end{proof}
\begin{remark}
\label{re:main3}
Condition ${2 \sigma_{V}\mu _e \!-\! {\sigma_V \alpha_h}/{2}  \!>\! 0 \!\!\iff\!\! \lambda_{\rm min}(\Lambda) \!>\! \alpha_h}$ requires the uncertainty estimation dynamics represented by $\lambda_{\rm min}(\Lambda)$ to be faster than the exponential rate of the CBF constraint, which is characterized by $\alpha_h$. For a given uncertainty estimator, one can select an $\alpha_h$ value that is sufficiently small in practical applications.
\end{remark}
\begin{remark}
\label{re:main2}
For Theorem~\ref{theo:main}, we require ${x_e(0) \!\in\! \C_V}$. This condition requires ${h(x_0) \!\geq\! \sigma_V \! V_e(e_0)}$, which is a more stringent condition than ${x_0 \!\in\! \C}$. This conservatism can be minimized by reducing the initial estimation error ${e(0)}$, which results in a smaller value of ${\sigma_V \! V_e(0)}$.
\end{remark}
\begin{remark}
\label{re:main1}
Safety condition \eqref{heineq} is not explicitly affected by the time-varying estimation error $e$, whereas \eqref{eq:thm4_set} depends on it. Hence, Theorem~\ref{theo:main} gives design flexibility even when the upper bound ${\delta_b}$ is unknown, as the upper bound of $\|e\|$, ${\bar{e}(t)}$, explicitly depends on ${\delta_b}$. 
\end{remark}

Finally, given a nominal controller ${\mathbf{k_d}}$, a CBF $h$, and parameters $\alpha_h$, ${\delta_L}$, ${\Lambda}$ for system \eqref{sysun}, to design an optimal robust safe controller ${{\bar{u}^*} \!\!=\! \mathbf{\bar{k}^*}(t,\!x),~\mathbf{\bar{k}^*} \!:\! \mathbb{R}^{+}_0 \!\! \times \!\! X \!\!\to\! \mathbb{R}^m}$ satisfying \eqref{heineq}, we define an uncertainty estimation-based robust safety filter, {UE-CBF-QP}:
\begin{align*}
\begin{array}{lllll}
{\mathbf{\bar{k}^*}(t,x)= \ }
\displaystyle  \argmin_{\bar{u} \in \mathbb{R}^m} \ \ \ \ {\|\bar{u}-\mathbf{k_d}(x)\|^2}  \\ [2mm]
 ~~~~~~~~~~~~~~~~~\textrm{s.t.} \ \ \ ~~~~~\eqref{heineq} \ \ {\rm and}\ \ \bar{u} - \hat{g}^\dagger\!(x) \hat{\Delta}(t) \in U,
\end{array}
\end{align*}
which is depicted in Fig.~\ref{fig:sector}.
\begin{figure}
	\centering
	\includegraphics[width=1\linewidth]{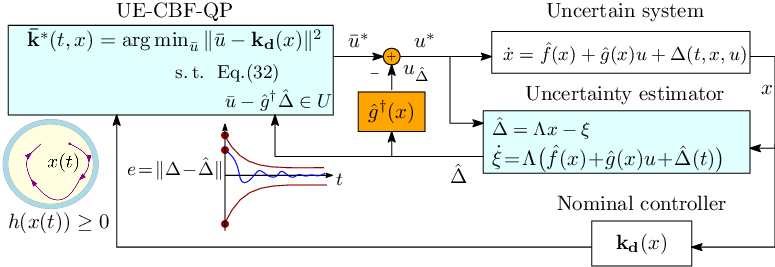}
	\caption{Block diagram of the uncertainty estimator-based safe control framework. Augmenting a given, and potentially unsafe, controller with an error-bounded uncertainty estimator and a safety filter guarantees that uncertain system states remain in a subset of the safe set.}
 \vspace{-2 mm}
	\label{fig:sector}
 \vskip -2 mm
\end{figure}

\subsection{Uncertainty Estimation-Based Safety with HOCBFs}
This subsection extends the uncertainty estimation-based robust safe control design method to HOCBFs in the presence of unmatched uncertainty under Assumption~\ref{as:DRD-IRD}.

Suppose ${\text{IRD} \!=\! r}$ and ${\text{DRD} \!=\! r \!-\! 1}$, and ${r \!\geq\! 2}$, the higher-order time derivatives of $h$ with respect to system \eqref{sysun} are given in \eqref{eq:hocbfder1}. We modify \eqref{eq:squphi} with respect to the uncertain system \eqref{sysun}, which results in:
\begin{align}
\begin{split}
\label{eq:secphi}
    &\phi_0 (x) \triangleq h(x); \ \phi_1 (x) \triangleq \dot{h}(x) + \alpha_1 ( {\phi}_{0} (x) ),
 \\
 & \phi_{r-1} (t,\! x,\! u) \triangleq 
 L_{\Delta} L_{\hat{f}}^{r-2} h(t,\! x,\! u) + \Psi (x), 
 \end{split}
\end{align}
where ${ \Psi (x) \!\triangleq\! L_{\hat{f}}^{r-1} h(x) \!+\!  \sum_{i= 1}^{r-1} L_{\hat{f}}^{i-1} ( \alpha_{r-i} \circ  {\phi}_{r-i-1} )(x)}$.
Note that $\Delta$ first appears in ${ \phi_{r-1}}$ as it depends on ${{h}^{r-1}}$, which depends on $\Delta$; see \eqref{eq:hocbfder1}. We also define the following sets and their intersections: 
\begin{equation}
\label{eq:hocbfset2}
    \C_i(t) \!\triangleq\! \{ (x,u) \!\in\! X \!\times\! U  \!:\! {\phi}_{i} (t,\! x,\! u) \!\geq\! 0 \} , \  
    \C(t) \!\triangleq\! {\textstyle \bigcap\limits_{i = 0}^{r -1}} \C_i(t)  ,
\end{equation}
which is a control input-dependent and time-varying set.
According to Theorem~\eqref{teo:hocbfdef}, safety is ensured with respect to set $\C(t)$ by the \textit{robust HOCBF condition} ${\forall x \!\in\! \C(t)}$: ${\Dot{\phi}_{r-1}(t,\! x,\! u)  \!\geq\! -\alpha_r ( {\phi}_{r-1}(t,\! x,\! u)  )}$, which can be expressed in terms of function $h$ as:
\begin{align}
\label{eq:hocbfproofmod}
      &  L_{\hat{f}}\Psi (x) \!+\! L_{\hat{g}} \Psi (x) u \!+\! \grad  \Psi (x)  {\Delta}(t,\! x,\! u) \!+\!  \grad L_{\hat{f}}^{r-2} h(x) \!~    \dot{\Delta}(t,\! x,\! u)  \nonumber \\
      &\!+\! \big ( \hat{f}(x) \!+\! \hat{g}(x) u \!+\! {\Delta}(t,\! x,\! u) \! \big )^\top  \mathcal{F}(x)  {\Delta}(t,\! x,\! u)  \nonumber \\ 
    &\!\geq\! -\alpha_r \big ( \Psi(x)  +  \grad L_{\hat{f}}^{r-2} h(x) {\Delta}(t,\! x,\! u) \big ), \ \forall x \in \C, 
\end{align}
where ${\alpha_r \!\in\! \mathcal{K}_{\infty, e}}$, ${ \mathcal{F}(x) \!\triangleq\! \derp{}{x} \big ( \grad L_{\hat{f}}^{r-2} h(x)  \big ) }$, with ${\mathcal{F} \!:\! X \!\to\! \mathbb{R}^{n \times n }}$. And, from Assumption~\ref{as1} we have
\begin{equation*}
     \grad L_{\hat{f}}^{r-2} h(x) \!~ \Dot{\Delta}(t,\! x,\! u)   \!\geq\! - \big \| \grad L_{\hat{f}}^{r-2} h(x) \big \| \delta_L,  \forall (\!t,\! x,\!u\!) \!\in\! \mathbb{R}^{+}_0 \!\times\! X \!\times\! U.
\end{equation*}

Like Theorem~\ref{th:met3}, we replace $\Delta$ with ${\hat{\Delta} \!+\! e}$, and $u$ with ${\bar{u} \!-\! \hat{g}^\dagger \hat{\Delta} }$ in \eqref{eq:hocbfproofmod} to integrate the uncertainty estimator into the robust safety constraint. This process results in the following error-dependent HOCBF condition similar to \eqref{pro:main_pro3}:
\begin{align}
\begin{split}
\label{eq:hocbfeineqfex}
& L_{\hat{f}}\Psi (x) +  L_{\hat{g}} \Psi (x) \big ( \bar{u} - \hat{g}^\dagger\!(x) \hat{\Delta}(t) \big ) +  \grad  \Psi (x)  \hat{\Delta}(t) \\
&+ \big ( \hat{f}(x) + \hat{g}(x) \big ( \bar{u} - \hat{g}^\dagger\!(x) \hat{\Delta}(t) \big ) + \hat{\Delta}(t)  \big )^\top \mathcal{F}(x)  \hat{\Delta}(t)  \\  
& \!-\!  \bar{e}(t) \big \| \big ( \hat{f}(x) + \hat{g}(x) \big ( \bar{u} - \hat{g}^\dagger\!(x) \hat{\Delta}(t) \big ) + \hat{\Delta}(t)  \big )^\top  \mathcal{F}(x)  \big \| \\
    &-  \big \| \grad L_{\hat{f}}^{r-2} h(x) \big \| \delta_L   - \Omega(t, x, e)  \\
   & \geq\! -\alpha_r  \big ( \Psi(x) \!+\!    \grad L_{\hat{f}}^{r-2} h(x) \hat{\Delta}(t) \!-\!\big \| \grad L_{\hat{f}}^{r-2} h(x) \big \| \bar{e}(t) \big ), 
\end{split}
\end{align}
where ${\Omega(t,x)  \!\!\triangleq\! \| \grad  \Psi (x)    \| \bar{e}(t) \!+\!\! \big \| \mathcal{F}(x) \big \| \big \|   \hat{\Delta}(t) \big \| \bar{e}(t)  \!\!+\!\! \big \|  \mathcal{F}(x) \big \| \bar{e}^2(t)}$. Note that we use the induced matrix 2-norm, which satisfies the submultiplicative property, in \eqref{eq:hocbfeineqfex} to bound the matrices.

\begin{theorem} \label{theo:rhocbfs}
Let ${h}$ be an HOCBF for the uncertain model \eqref{sysun} on set ${\C(t)}$, constructed in \eqref{eq:hocbfset2}. Let uncertainty ${\Delta}$, of the form \eqref{eq:munmdel},  satisfy Assumptions~\ref{as1} and ~\ref{as:DRD-IRD}. Given an uncertainty estimator with an associated Lyapunov function $V_e$ that satisfies \eqref{eq:er_iss}, and a feedback control law ${{u} \!=\!  \bar{u} \!-\! \hat{g}^\dagger\!(x) \hat{\Delta}(t)  }$, any Lipschitz continuous controller ${\bar{u} \!=\! \mathbf{\bar{k}}(x)}$ satisfying \eqref{eq:hocbfeineqfex} renders set ${ \C(t) }$ forward invariant.
\end{theorem}
\begin{proof}
\label{pro:thecone}
We follow the same steps described in the proof of Theorem~\ref{th:met3}. We replace the unknown terms with their lower bounds to obtain a sufficient HOCBF condition. As inequality in \eqref{eq:hocbfeineqfex} ${\!\!\implies\!\!}$ \eqref{eq:hocbfproofmod}, we have that system \eqref{sysun} is safe with respect to set $\C(t)$ by Theorem~\ref{teo:hocbfdef}.
\end{proof}
\vskip -0.1 true in
Then, given a nominal controller ${\mathbf{k_d}}$, to design a controller ${{\bar{u}^*} \!\!=\! \mathbf{\bar{k}^*}(t,\!x),~\mathbf{\bar{k}^*} \!:\! \mathbb{R}^{+}_0 \!\! \times \!\! X \!\!\to\! \mathbb{R}^m}$ satisfying \eqref{eq:hocbfeineqfex}, we define an uncertainty estimation-based robust {UE-HOCBF safety filter}:
\begin{align*}
\begin{array}{lll}
{\mathbf{\bar{k}^*}(t,x)= \ }
\displaystyle  \argmin_{\bar{u} \in \mathbb{R}^m} \ \ \ \ {\|\bar{u}-\mathbf{k_d}(x)\|^2}  \\ [4mm]
~~~\textrm{s.t.} \ \  \tilde{c}(t, x, u) \leq \tilde{a}(t, x) \bar{u} +  \tilde{b}(t, x) \\ [1mm]
~~~~~~~~~~~ \bar{u} - \hat{g}^\dagger\!(x) \hat{\Delta}(t) \in U,
\end{array}
\end{align*}
where
\begin{align}
\begin{split}
\label{eq:qcqpparam}
&\!\!\tilde{a}(t, \!x) \triangleq L_{\hat{g}} \Psi (x) \!+\!  \big (  \hat{g}^\top\!(x)\mathcal{F}(x) \hat{\Delta}(t)   \big )^\top  , \\
&\!\!\tilde{b}(t, \!x) \triangleq 
 L_{\hat{f}}\Psi (x) \!+\!  L_{\hat{g}} \Psi (x) \big ( \!- \hat{g}^\dagger\!(x) \hat{\Delta}(t) \big ) \!+\!  \grad  \Psi (x)  \hat{\Delta}(t) \\
&\!\!+ \big ( \hat{f}(x) + \hat{g}(x) \big ( - \hat{g}^\dagger\!(x) \hat{\Delta}(t) \big ) + \hat{\Delta}(t)  \big )^\top \mathcal{F}(x)  \hat{\Delta}(t)  \\  
&\!\!-  \big \| \grad L_{\hat{f}}^{r-2} h(x) \big \| \delta_L   - \Omega(t, x, e)  \\
& \!\! + \alpha_r \! \big ( \Psi(x) \!+\!    \grad L_{\hat{f}}^{r-2} h(x) \hat{\Delta}(t) \!-\!\big \| \grad L_{\hat{f}}^{r-2} h(x) \big \| \bar{e}(t) \big ) \\
&\!\! \tilde{c}(t, \!x, \!u) \!\triangleq\! \bar{e}(t) \big \| \big ( \hat{f}(x) \!+\! \hat{g}(x) \big ( \bar{u} \!-\! \hat{g}^\dagger\!(x) \hat{\Delta}(t) \big ) \!+\! \hat{\Delta}(t)  \big )^\top \! \mathcal{F}(x)   \big \|. 
\end{split}
\end{align}
The sufficient HOCBF condition \eqref{eq:hocbfeineqfex} is not affine in the control $\bar{u}$. Thus, the UE-HOCBF optimization problem is not a QP. However, we show next how to integrate the HOCBF condition into a second-order cone program (SOCP). To make the objective function linear, we convert the quadratic cost to \textit{epigraph form} by adding a new variable: ${\rho \!\in\! \mathbb{R}}$:
\begin{align*}
\begin{array}{llll}
{[\mathbf{\bar{k}^*}(t,x),~ \rho^*]= \ }
\displaystyle  \argmin_{\bar{u} \in \mathbb{R}^m,~\rho \in \mathbb{R}} \ \ \ \ {\rho}  \\ [4mm]
~~~~\textrm{s.t.} \ \  \tilde{c}(t, x, u) \leq \tilde{a}(t, x) \bar{u} +  \tilde{b}(t, x) \\ [1mm]
~~~~~~~~~\|\bar{u} -\mathbf{k_d}(x) \|^2 \leq \rho \\ [1mm]
~~~~~~~~~\bar{u} - \hat{g}^\dagger\!(x) \hat{\Delta}(t) \in U,
\end{array}
\end{align*}
where the rotated second-order cone constraint ${\|\bar{u} -\mathbf{k_d}(x) \|^2 \leq \rho}$ can be reformulated in the form of a standard second-order cone:
\begin{equation}
    \label{eq:socprot}
    \left \| \begin{bmatrix}
        2 (\bar{u} -\mathbf{k_d}(x)) \\
        \rho - 1
    \end{bmatrix}\right \| \leq \rho + 1 .
\end{equation}
The UE-HOCBF-SOCP is the solution to the optimization problem above with the second constraint replaced by \eqref{eq:socprot}. The UE-HOCBF-SOCP guarantees forward invariance of ${\C(t)}$.
\begin{remark}
    \label{re:Hocbf-qp}
If ${ ( \hat{g}(x) \bar{u} )^\top \! \mathcal{F}(x) \!\equiv\! 0}$, the left-hand side of the inequality in the {UE-HOCBF safety filter} becomes independent of $\bar{u}$, and the optimization problem reduces to a QP. More specifically, suppose that ${\mathcal{F}(x) \!\equiv\! 0}$, then the constraint can be integrated into an uncertainty estimation-based QP, UE-HOCBF-QP. Furthermore, if ${ ( \hat{g}(x) \bar{u} )^\top \! \mathcal{F}(x) \!\equiv\! 0}$, we can extend Theorem~\ref{theo:main} to HOCBFs. In order to achieve this objective, we need to modify ${\phi_{r-1} }$ and its 0-superlevel set to provide robustness against estimation error $e$, analogous to \eqref{hecbf}:
\begin{align*}
   \phi_{r-1}^V (t, x_e, u)  \triangleq \phi_{r-1} (t,\! x,\! u) \!-\! \sigma_V V_e(e) , ~ \sigma_V \!\in\! \mathbb{R}^+ , \\
    \C_{r-1}^V(t) \!\triangleq\! \big\{ (x_e,\! u) \!\in\! X \!\times\! \mathbb{R}^{n} \!\times\! U  ~|~ \phi_{r-1}^V (t,\! x_e,\! u) \!\geq\! 0 \big\} , 
\end{align*}
which is a subset of ${\C_{r-1}(t)}$. Then, one can easily follow the same steps in Theorem~\ref{theo:main} and its proof.
\end{remark}

\begin{remark}
\label{re:hmdorhocbf}
To compensate for the mismatch between the nominal model and the uncertain system dynamics, we subtract the estimated uncertainty from the UE-CBF-QP and UE-HOCBF-SOCP outputs, as given in equation \eqref{eq:u_hat}. This modification improves closed-loop control performance, but it is not related to safety. However, since it is a part of feedback control, it appears in the CBF constraints, as seen from the QPs. Therefore, one can remove ${-\hat{g}^\dagger\!(x) \hat{\Delta}(t)}$ from the composite control law \eqref{eq:u_hat}. In this scenario, we only need to replace  ${\mathbf{\bar{k}}(x) \!-\! \hat{g}^\dagger\!(x) \hat{\Delta}(t)}$ with ${\mathbf{\bar{k}}(x)}$ between equations \eqref{eq:u_hat} and \eqref{eq:qcqpparam}. 
\end{remark}

\section{Simulations And Hardware Experiments}
\label{ex:sim_hard}

\subsection{Elastic Actuator with Matched and Unmatched Uncertainties}
\label{ex:EA}
We first demonstrate our robust HOCBF method on a simulated elastic actuator problem characterized by both full-state matched and unmatched uncertainty.  

Consider a flexible-joint mechanism, adopted from \cite{slotine1991applied} (Example 6.10), whose dynamics take the form \eqref{eq:sysune1} with the components:
\begin{align*}
\begin{split}
&\hat{f}(x)  \!=\!\! \begin{bmatrix}
x_2 \!\!\!&
{(m \bar{g} L \sin{{x}_1} \!-\! k ({x}_1 \!-\! {x}_3 ))}/{I_L}  \!\!\!&
x_4 \!\!\!&
{k} ({x}_1 \!-\! {x}_3 ) / {J_m}
\end{bmatrix}^\top,
\\
&\hat{g}(x) \!=\!\!
\begin{bmatrix}
0 &
0 &
0 &
{1}/{J_m} 
\end{bmatrix} ,
\\
&\Delta(x, u) \!=\!\!
\begin{bmatrix} 
0 \!\!\!& 
-\frac{1}{3} \hat{f}_2(x) - 0.7 \!\!\!&
0 \!\!\!&
\frac{1}{4} \hat{f}_4(x) \!-\! 0.2 + \frac{u}{4 J_m} 
\end{bmatrix}, 
\end{split}
\end{align*}
where ${\Delta(x, u) \!\triangleq\! [0~\Delta_2(x)~0~\Delta_4(x, u)]^\top}$, ${\hat{f}_2}$ and ${\hat{f}_4}$ are the second and fourth components of $\hat{f}$, ${{x}_1, {x}_3}$ are the joint angles of the actuator, ${J_m \!=\! 0.1 ~ \text{[kg}\text{m}^2]}$ is the motor rotational inertia, ${k \!=\! 0.25 ~ \text{[Nm/{rad}]}}$ is the torsional spring stiffness, ${m \!=\! 0.5~ \text{[kg]}}$ is the load's mass, ${I_L \!=\! 0.5~\text{[kg}\text{m}^2]}$ is the load rotational inertia, ${L \!=\! 0.04 ~ \text{[{m}]}}$ is the eccentricity of the load's center of mass, $\bar{g}$ is the gravitational constant, and $u$ is the motor torque. The uncertainty takes the form \eqref{eq:sysune1}, and ${\Delta_2(x) \!\perp\! \Ima(\hat{g}(x)) }$. The control input set is ${ U \!\triangleq\! \left\{ u \!\in\!  \mathbb{R} \!:\! -u_{\rm max} \!\leq\! u \!\leq\!  u_{\rm max} \right\}}$, where ${u_{\rm max} \!=\! 0.2 \bar{g}}$. We obtain ${\hat{g}^\dagger(x) \!=\! [0~0~0~J_m]^\top}$, which yields the full-state matched and unmatched uncertainty using \eqref{syscon} as ${\Delta_m(x, u)  \!=\! \Delta_4}$, ${\Delta_u(x)  \!=\!  \Delta_2}$. 

To synthesize a nominal feedback controller, we choose a \textit{control Lyapunov function} (CLF): ${V(x) \!=\! (x_4 \!-\! x_{4d})^2}$, ${x_{4d} \!=\! 0}$, which embodies the tracking requirement for state ${x_4}$, where ${x_{4d}}$ is the desired trajectory. This CLF choice yields the min-norm controller: 
\begin{equation*}
\mathbf{k_d}(x) \!=\! \begin{cases} \!\!-\dfrac{\overbrace{L_f V (x) \!+\! \lambda_V V (x)}^{\triangleq \varsigma(x)}}{\|L_g V (x)\|^2}L_g V^\top (x) & \!\!\!\!\text { if } \varsigma(x)\!>\!0 \\ 0 & \!\!\!\!\text { if } \varsigma(x) \!\leq \!0\end{cases}
\end{equation*}
that is the solution of the CLF-QP \cite{jankovic2018robust} with ${\lambda_V \!=\! 10}$. 

We choose a CBF, ${h(x) \!=\! x_2 \!-\! c_q x_3}$, where ${c_q \!=\! -2}$, and ${\alpha_1 (h) \!=\! 2 h, \alpha_2 (h) \!=\! 2 h}$. The CBF ${h}$ has IRD of 2, i.e., ${r \!=\! 2}$ in \eqref{eq:hocbfder1}, and DRD of 1. This uncertainty structure and CBF choice satisfy Assumption~\ref{as:DRD-IRD}. We use an HOCBF framework. From \eqref{eq:secphi} and \eqref{eq:hocbfproofmod} we have ${\phi_0(x)  \!=\! x_2 \!-\! c_q x_3}$, ${\phi_{1}(x)  \!\!=\!\!  \Delta_2(x) \!+\!{ {m \bar{g} L \sin{{x}_1}  \!-\! k ({x}_1 \!-\! {x}_3 )}   }/{I_L} \!-\! c_q x_4 \!+\! \alpha_1( x_2 \!-\! c_q x_3 \!)}$, and ${\mathcal{F}(x) \!\equiv\! 0}$. Since ${\mathcal{F}(x) \!\equiv\! 0}$, we use an UE-HOCBF-QP as a safety filter. Substituting ${\phi_0,~\phi_1}$ into \eqref{eq:qcqpparam} yields an affine HOCBF condition for the actuator system for some ${\alpha_2 \!\in\! \mathcal{K}_{\infty, e}  }$. The initial state values are  ${x(0) \!=\! [0~\!0.5~\!0~\!\!-\!0.2]^\top}$. To obtain bounding constants ${\delta_L, \delta_b}$ for the problem's uncertainty, we simulate numerous scenarios, resulting in ${\delta_L \!=\! 0.6, \delta_b \!=\! 1.4}$. We choose ${\Lambda \!=\! \diag(5, 5, 5, 5)}$ resulting in ${ \mu_e \!=\! 1.25}$, and  ${H \!=\! \diag(1, 1, 1, 1)}$ to solve the Lyapunov equation.
\begin{figure}
	\centering
	\includegraphics[width=1.0\linewidth]{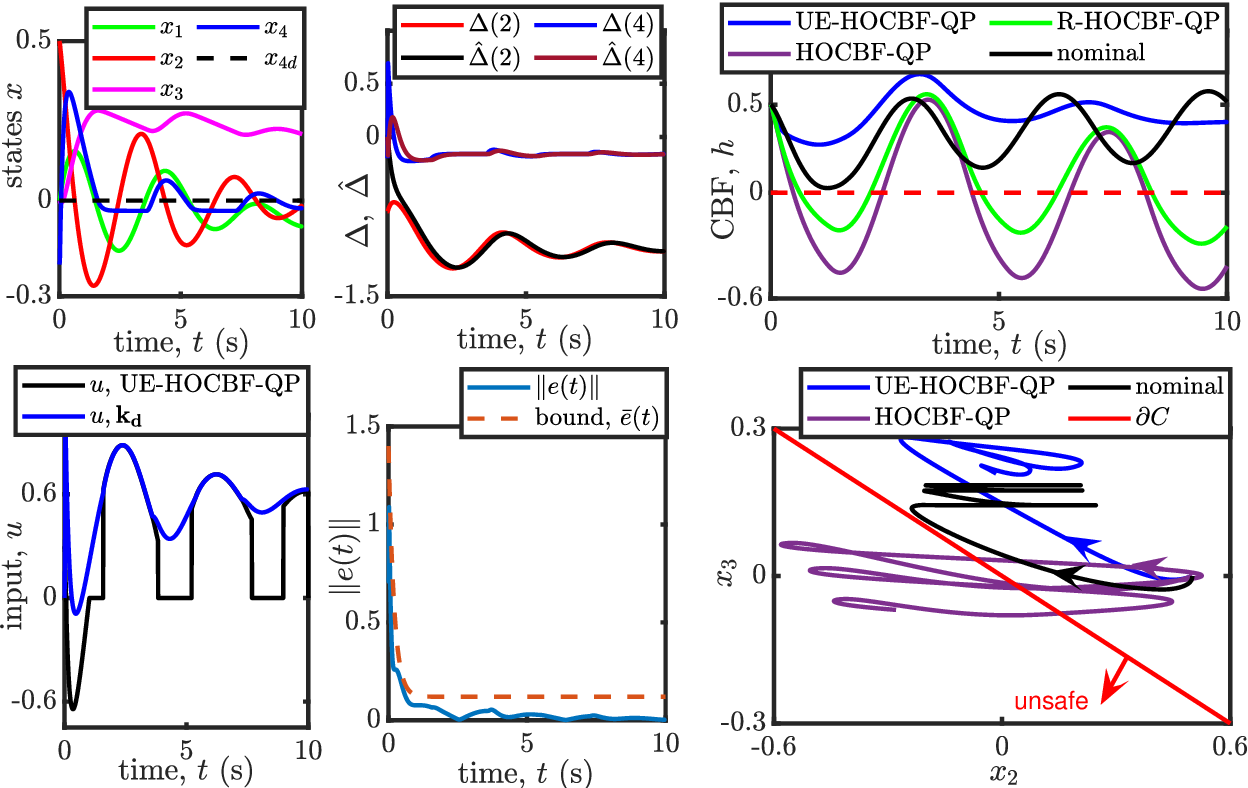}
	\caption{Simulations of the uncertain elastic actuator. \textbf{(Left)} (Top) State and reference tracking of the uncertain system operating under the UE-HOCBF-QP. (Bottom) Values of the safe control input, $u$ safe, and the nominal input, $u$ nominal, vs. time.
    \textbf{(Center)} (Top) The proposed estimator can effectively estimate actual model uncertainties within the quantified bounds. Because of ${\hat{\Delta}(0) \!=\! \bf 0}$ and ${\Delta_1 \!\equiv\! 0,~\Delta_3 \!\equiv\! 0}$, the estimator output is zero for these uncertainty components. (Bottom) The estimation error with the theoretical bound which is satisfied.
    \textbf{(Right)} (Top) The value of CBF $h$ vs. time and trajectory of the system with the boundary of the safe set. (Bottom) The proposed UE-HOCBF-QP safety filter maintains safety in the presence of unmodeled dynamics. While HOCBF-QP satisfies safety for the nominal system, it leaves the safe set with the uncertain system. Under the robustified HOCBF (R-HOCBF) approach \cite{tan2021high}, the system leaves the safe set due to ignoring the unmatched uncertainty, which leads to safety violations.} 
	\vspace{-4mm}
	\label{fig:result_EA}
\end{figure}

We compared the proposed method to the robustified HOCBF \cite{tan2021high} (Remark 3) with a known time-invariant upper bound of the matched uncertainty ${\Delta_m}$ such that ${ \|\Delta_m(x,u)\| \!<\! \Bar{\Delta}_m  }$, ${\forall (x, u) \in X \times U}$. However, this method ignores the unmatched uncertainty ${\Delta}_u$ and thus does not guarantee safety in the presence of unmatched disturbances, as it relies on the assumption that ${\text{IRD} \!=\! \text{DRD} \!=\! r}$. As discussed in Section~\ref{sec:intro}, our method is, to the best of our knowledge, the first to ensure robust safety against both matched and unmatched uncertainties via HOCBFs.   

Data from the simulated system (100 Hz sampling rate) are presented in Fig.~\ref{fig:result_EA}. With the HOCBF-QP controller, the nominal system (the elastic actuatorwithout uncertainty) remains safe. However, safety violations occur in the presence of unmodelled system dynamics. Additionally, the robustified HOCBF method \cite{tan2021high} does not guarantee safety, as it only considers matched uncertainty. Note that the HOCBF-QP results in an unsafe system, whereas the system remains safe with our proposed method. Fig.~\ref{fig:result_EA} also demonstrates that the proposed uncertainty estimator effectively estimates actual modeling uncertainties within the quantified bounds.

\subsection{Experimental Validation on a Mobile Robot}
\label{sec:exper}
Real-world safety-critical control systems suffer from incompletely modeled uncertainties that may degrade the safety guarantees of controllers designed from nominal models. To address this robustness problem, we must consider the actual system models when synthesizing safe controllers, as exemplified in the following robotics example.

Tracked mobile robots are difficult to model exactly. Simplified models,  such as the unicycle model, are commonly used in the design of tracked vehicle controllers.
\begin{figure}
\centering
\vspace{-4mm}
\subfloat{%
   \includegraphics[width=0.405\columnwidth]{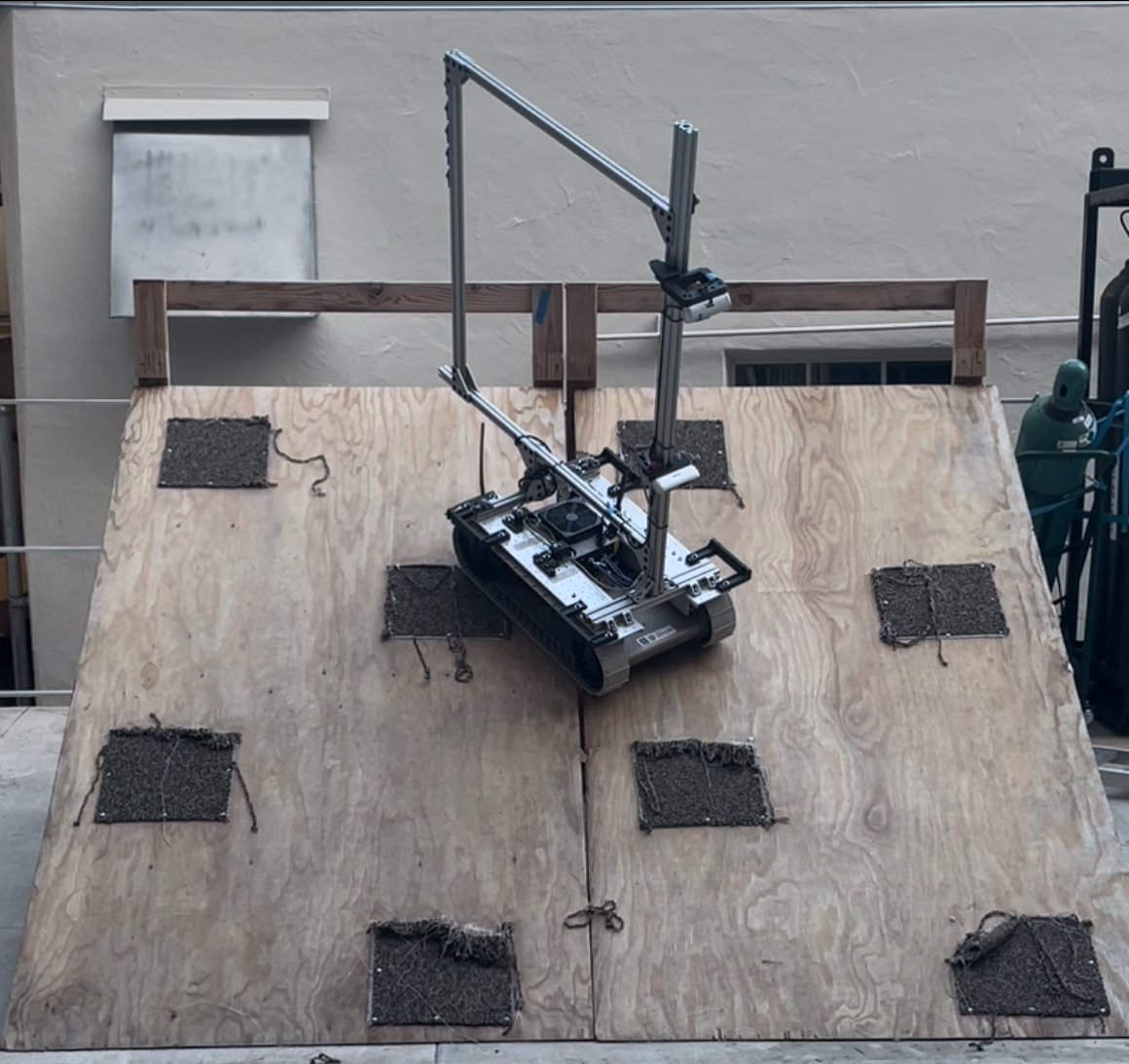}%
   \label{fig:robot-sub1}}
\hfill
\subfloat{%
   \includegraphics[width=0.58\columnwidth]{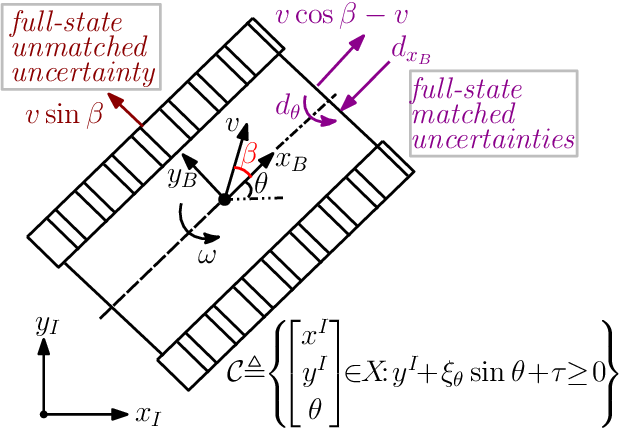}%
   \label{fig:robot-sub2}}
\caption{Photo of a tracked mobile robot on an inclined surface during one experiment. The higher friction carpet patches on the slope result in slip uncertainty when only one track lies on a patch \textbf{(Left)}. The geometry of full-state matched and unmatched uncertainties due to skidding and slipping, and description of the safe set ${\C}$ for Example~\ref{sec:exper} \textbf{(Right)}.}
\vskip -3mm
\label{fig:robot6d}
\vskip -3mm
\end{figure}
A more detailed model lumps all of the unmodeled track-terrain interaction effects intoa time-varying {\em slip angle} $\beta$, a {\em yaw rate uncertainty} $d_{\theta}$, and a {\em longitudinal slip velocity} ${d_{x_B}}$ that models slip between the tracksand the terrain (see Fig.~\ref{fig:robot6d}):
\begin{equation}
\label{eq:unicyle}
        \underbrace{
   \begin{bmatrix}
    \dot{x}^I \\
    \dot{y}^I \\
     \dot{\theta } 
    \end{bmatrix}}_{\dot{x}}
     \!=\! 
     \underbrace{
     \begin{bmatrix}
    \cos{\theta} & 0    \\
     \sin{\theta} & 0  \\
    0 & 1 
    \end{bmatrix}}_{\hat{g}(x)}
    \underbrace{
    \begin{bmatrix}
        v \\ \omega
   \end{bmatrix}}_{u} 
    \!+\! \underbrace{\begin{bmatrix}
    \!v (\cos\gamma \!-\! \cos\theta) \!-\!d_{x_B} \cos{\theta}  \\
     \!v (\sin\gamma \!-\! \sin\theta) \!-\!d_{x_B} \sin{\theta} \\
      \!d_{\theta}
     \end{bmatrix}}_{\Delta(t,\! x,\! u)} ,
\end{equation}
where ${ [ {x}^I ~ {y}^I ]^\top}$ is the vehicle's planar position with respect to inertial frame ${I}$, $\theta$ is itsyaw angle, ${  v }$ is its linear velocity, ${\omega }$ is the angular velocity, and ${\gamma\!\triangleq\! \theta \!+\! \beta}$. Note that \eqref{eq:unicyle} is a nonlinear control affine system with time-varying, state, and input-dependent uncertainty, like  \eqref{sysun}.  This more advanced model, which still simplifies many physical effects, and whose parametersvary with vehicle speed and terrain characteristics, allows us to quantitatively analyze the effects of disturbances. In Fig.~\ref{fig:robot6d}, note that slipcomponent ${v \cos{\beta} \!-\! v}$ anddisturbance ${d_{x_B}}$ are parallel to the ${x_B}$-axis. These uncertainties, which affect the robot's linear velocity,can be compensated by control $v$. Thus, they are classified as full-state matched uncertainty/disturbance. Conversely, the second slip component, ${v \sin{\beta}}$, is perpendicular to ${x_B}$. It cannot be compensated by inputs ${v}$ or ${w}$; therefore, ${v \sin{\beta}}$ is a full-state unmatched uncertainty. The full-state matched and unmatched uncertainty are obtained using \eqref{eq:matched_fin} and \eqref{eq:unmatched_fin}:
 \begin{align}
 \begin{split}
 \label{eq:matched_unif}
    \Delta_m(t,\! x,\! u) &=
    \begin{bmatrix}
    v \cos{\theta} ( \cos{\beta } - 1  ) -d_{x_B} \cos{\theta} \\
     v \sin{\theta} ( \cos{\beta } - 1  ) -d_{x_B} \sin{\theta} \\
     d_{\theta} 
     \end{bmatrix}, \\
     \Delta_u(t,\! x,\! u) &= \begin{bmatrix}
   - v \sin{\theta}  \sin{\beta }  & 
     v \cos{\theta}  \sin{\beta }  &
     0 
     \end{bmatrix}^\top,
 \end{split}
 \end{align}
where $\Delta_u$ is the observed value of vector ${v \sin{\beta}}$ in the inertial frame, and representslateral vehicle motion due to improperly modeled slip.   

Next, we consider how the uncertainty $\Delta$ in \eqref{eq:unicyle} impacts the CBF constraint. Assume that the mobile robot must maintain a safe distance ${\tau\! \in\! \mathbb{R}^+_0}$ froma wall aligned in the ${y}^I$ direction: $h(x) \!=\! {y}^I  \!+\! \xi_{\theta}  \sin{\theta} \!+\! \tau$, where ${\xi_{\theta}\! \in\! \mathbb{R}^+_0}$determines the effect of the heading angle on safety (see Fig.~\ref{fig:robot6d}). This CBF choice results in $L_{\hat{f}} h(x) \!=\! 0$, $L_{\hat{g}} h(x) \!=\! [ \sin{\theta} ~ \xi_{\theta}\cos{\theta}]^\top$, $\grad h \Delta(t\!,x,\! u) \!=\! v ( \sin{\gamma}  \!-\! \sin{\theta} ) \!-\!d_{x_B}\! \sin{\theta} \!+\! d_{\theta} \xi_{\theta} \cos{\theta}$, and the CBF ${h}$ has ${\text{DRD} \!=\! \text{IRD} \!=\! 1}$.

Note that $h$ is a CBF for the nominal system as ${ L_{\hat{g}} h(x)  \!\neq\! 0}$, ${\forall x\! \in\! X}$; thus, we have ${v \sin{\theta} \!+\! \omega \xi_{\theta}\cos{\theta} \!\geq\! - \alpha ({y}^I \!+\! \xi_{\theta}  \sin{\theta} \!+\! \tau)}$, ${\forall x \!\in\! \C}$, which is affine in the control input vector $u$. To analyze the effect of uncertainty on safety, consider the CBF \eqref{eq:hdot} under the uncertain unicycle model \eqref{eq:unicyle}:
\begin{align}
\begin{split}
\label{eq:cbfununi}
&v \sin{\theta} + \omega \xi_{\theta}\cos{\theta} + \alpha ({y}^I +  \xi_{\theta}  \sin{\theta} + \tau) \!\geq\! 0 \\
&~~~~~\geq -v ( \sin{\gamma } - \sin{\theta} )  + d_{x_B} \sin{\theta} - d_{\theta} \xi_{\theta} \cos{\theta}.
\end{split}
\end{align}
The left-hand side of \eqref{eq:cbfununi} is satisfied since $h$ is a CBF for the nominal system. The right-hand side of \eqref{eq:cbfununi} quantifies the effect of uncertainty on safety: when it is positive (e.g., if ${v \!>\! 0}$,  ${-\pi/4 \!<\! \theta \!<\! 0}$,  ${-\pi/4 \!<\! \beta \!<\! 0}$, ${d_{x_B} \!\leq\! 0}$, and ${d_{\theta} \!\leq\! 0}$, then the right-hand term will be positive), safety may be violated. To address this robustness problem, we implemented our {UE-CBF-QP}.

We conducted an experimental test on a ground mobile robot (Fig.~\ref{fig:robot6d}). The test vehicle and its onboard computation are described inSection IV of \cite{janwani2023learning}. The 25$^{\circ}$ inclined test surface caused both slip-induced matched and unmatched uncertainty. We used the CBF $h$ with ${\tau \!=\! 1.0~[m]}$, ${\xi_{\theta} \!=\! 0.1}$ and the nominal controller ${\mathbf{k_d}(x) \!=\! [
        K_v  d_g  ~
        K_{\omega}  {y_g \!-\! y^\mathcal{I}}\!/\!{d_g} \!-\! K_{\omega} \sin{\theta} 
]^\top}$, with ${K_v \!=\! 1,\!~K_{\omega} \!=\! 1.2}$, are the controller gains, ${x_g \!=\! 1~[m],\!~y_g \!=\! 1~[m] }$ is the robot's goal position, and ${d_g \!\triangleq\! \|  x_g \!-\! x^\mathcal{I}, \!~ y_g \!-\!y^\mathcal{I} \| }$. The robot's control input setis given by: ${U \triangleq \big\{ u \!\in\!  \mathbb{R}^2 \!:\! -u_{\rm max} \!\leq\! u \!\leq\! u_{\rm max} \big\}}$, where ${u_{\rm max} \!\triangleq\! [0.5 ~ 0.5]^\top }$. Thesafety-critical control parametersare: ${\alpha_h \!=\! 1}$, ${\sigma_{V} \!=\! 1 }$, $\Lambda \!=\! \diag(4, 4, 4)$, ${\delta_L \!=\! 0.5}$, ${\delta_b \!=\! 1.5}$, and ${H \!=\! \diag(8, 8, 8)}$. The control loops operatedat a 50 Hzrate.

The robot's reference trajectory is chosen to violate the safety constraint when using the nominal controller. Under a CBF-QP approach, the robot's path leaves the safe set due to the unmodeled uncertainty--see Fig.~\ref{fig:GVRrobot} \textbf{(Right)} (Middle) and (Bottom). Note that the uncertainty vector is the difference between the discrete-time derivative of the hardware system's flow mapand the nominal robot model. Uthe robot states remainwithin the 0-superlevel set of the chosen CBF. Fig.~\ref{fig:GVRrobot} \textbf{(Center)} shows that our approach can estimate the model uncertainty, which affects the safety constraints. 
\begin{figure}
	\centering
	\includegraphics[width = \linewidth]{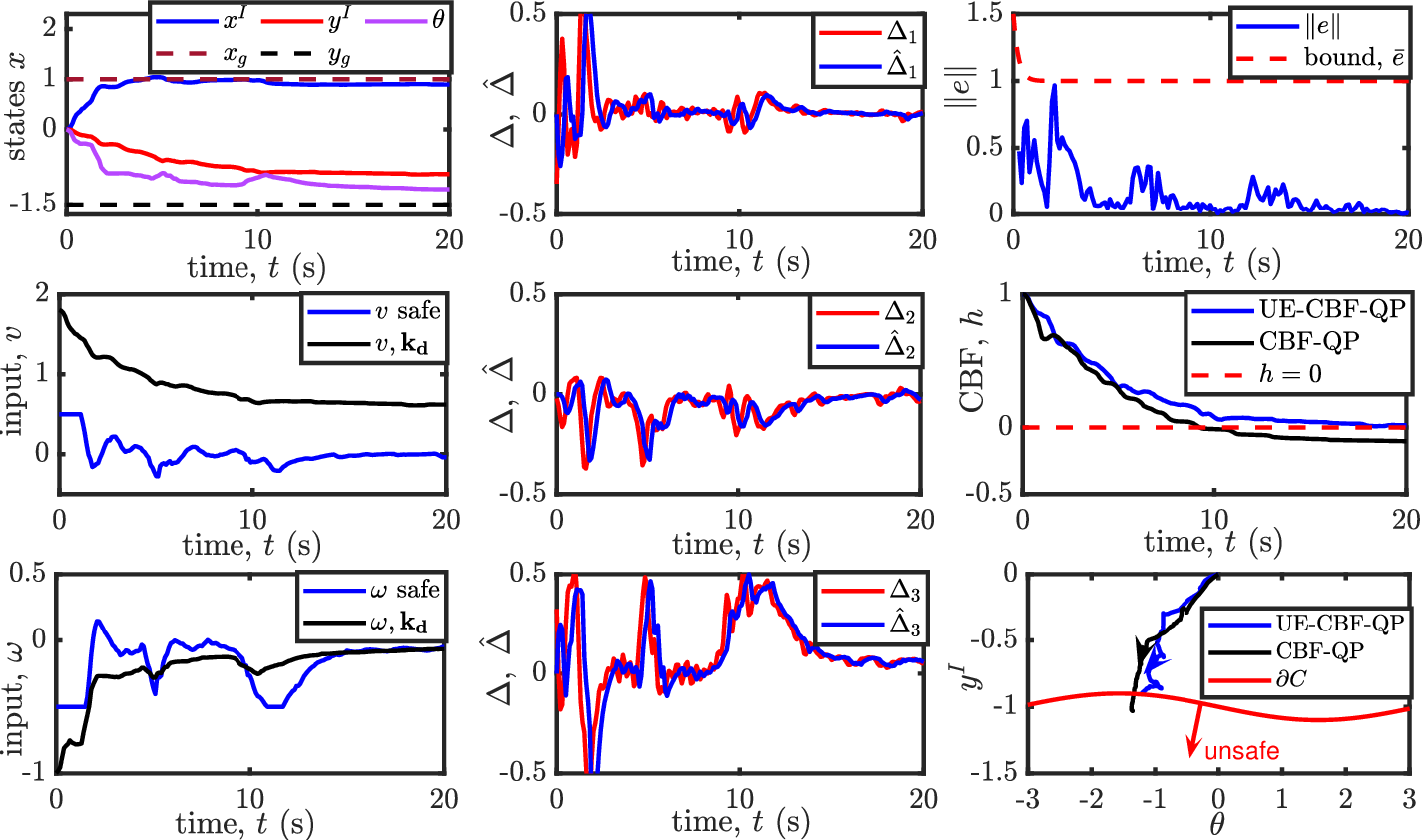}
\vskip -0.1 true in
\caption{Experimental results for the mobile robot example of Fig.~\ref{fig:robot6d}. 
\textbf{(Left)} (Top) State trajectories (${x^I~[m],~y^I~[m],~\theta~[rad]}$) and goal position of a robot operating under the UE-CBF-QP. (Middle) and (Bottom) The values of the safe linear velocity input, $v~[m/s]$ safe, the angular velocity input, $\omega~[rad/s]$ safe, and the nominal control inputs, $v$ nominal, $\omega$ nominal, vs. time. 
\textbf{(Center)} The uncertainty estimator tracksmodeing inaccuracies, as shown by the first and last components of the uncertainty vector,${\Delta_1}$ and ${\Delta_3}$. 
\textbf{(Right)} (Top) The uncertainty estimation error stays within the theoretical bounds. (Middle) The value of CBF $h$ vs. time. (Bottom) System trajectory is plotted with the safe set boundary. The proposed UE-CBF-QP-based safety filter maintains a safe distance from an edge on a slope in the presence of unmodeed dynamics. However, the CBF-QP controlled robot leaves the safe set due to the slope-induced uncertainty.}
\vspace{-4 mm}
\label{fig:GVRrobot}
\end{figure} 

\section{Conclusions and Future Work}
\label{sec:conc}
We presented a method to synthesize robust controllers with formal safety guarantees for nonlinear systems that experience both matched and unmatched uncertainty. A novel uncertainty estimator is applied to state and input-dependent uncertainty, leading to bound in the CBF framework. By integrating the estimator with CBFs and HOCBFs, robust safety conditions were derived. The utility of ourframeworks was experimentally demonstrated by controlling a tracked mobile robot on a slope while using a reduced-order nominal model.

Future work will estimate the bound $\delta_L$ via data-driven techniques. Additionally, our methods can also consider state estimation errors.

\begin{spacing}{0.865}
\bibliographystyle{IEEEtran}
\section*{References}
\vspace{-4 mm}
\bibliography{Refs}
\vspace{-0 mm}
\end{spacing}

\end{document}